\documentclass[journal,11pt]{IEEEtran}

\usepackage[a4paper, margin = 1.05in]{geometry}

\usepackage{setspace}
\usepackage[flushleft]{threeparttable}
\usepackage{longtable,booktabs,threeparttablex}

\usepackage[boxed,vlined,ruled]{algorithm2e}
\usepackage[utf8]{inputenc} 
\usepackage[T1]{fontenc}
\usepackage{url}
\usepackage{ifthen}

\usepackage[usenames,dvipsnames,table]{xcolor}

\usepackage{array,graphicx}
\usepackage{amsfonts,cite,euscript}
\usepackage[cmex10]{amsmath}
\usepackage[amsmath]{empheq}
\usepackage{cases,balance,cite, lineno}
\usepackage{booktabs,mathrsfs,bbm,amssymb,soul,lipsum}
\usepackage{enumitem,amsthm}

\graphicspath{{./Figures/}}

\theoremstyle{plain}
\newtheorem*{theorem*}{Theorem}
\newtheorem{theorem}{Theorem}

\newtheorem{lemma}{Lemma}

\newtheorem{proposition}{Proposition}

\onecolumn

\def\sr{\textsc{Sr-Adc}}

\def\fo{f_{{\sf{Out}}}}
\def\fin{f_{\sf{In}}}
\def\etal{\emph{et al.~}}

\newcommand{\PW}[1]{\mathsf{PW}_{#1}}

\renewcommand\leq{\leqslant}

\renewcommand\geq{\geqslant}

\newcommand{\ud}[2]{\mathrm{unif}\left(#1,#2\right)}

\def\YN{{\color{black}y _{\eta}}}

\newcommand{\DA}[1]{\stackrel{(\mathrm{#1})}{=}}
\def\l{\left(}
\def\r{\right)}
\def\T{{\color{black}\mathrm{T}}}

\def\TE{{\color{black}\mathrm{T}_\epsilon}}
\def\TS{{\color{black}\mathrm{T}_{\sf{Shannon}}}}
\def\TUS{{\color{black}\mathrm{T}_{\sf{US}}}}

\def\Z{\in \mathbb{Z}}
\def\R{\in \mathbb{R}}

\def\DE{\stackrel{\rm{def}}{=}}

\def\fo{f_{{\sf{Out}}}}
\def\fin{f_{\sf{In}}}

\def\B{\beta_g}
\def\DR{\mathsf{DR}}

\def\S{\mathsf{S}}
\def\bo{b_0}

\newcommand{\fig}[1]{Fig.~\ref{#1}}
\newcommand{\md}[2]{#1~\rm{mod}~#2}

\newcommand{\ab}[1]{{\color{black} #1}}
\newcommand{\aq}[1]{{\color{black} #1}}
\newcommand{\vi}[1]{{\color{black} #1}}

\newcommand{\MO}[1]{\mathscr{M}_\lambda\l #1 \r}
\newcommand{\VO}[1]{\varepsilon_{#1}}

\newcommand{\VOB}[1]{\color{black}\Delta^{N}{\varepsilon}_ {#1}}

\newcommand{\DN}[2]{\color{black}{\Delta^{#1}}#2}
\newcommand{\ND}[1]{\color{black}{\Delta^{N}{#1}}}

\newcommand{\ft}[1]{\left[\kern-0.15em\left[#1\right]\kern-0.15em\right]}
\newcommand{\fe}[1]{\left[\kern-0.30em\left[#1\right]\kern-0.30em\right]}
\newcommand{\flr}[1]{\left\lfloor #1 \right\rfloor}
\newcommand{\DL}[1]{\stackrel{(\mathrm{#1})}{=}}
\newcommand{\EQc}[1]{\stackrel{(\ref{#1})}{=}}

\newcommand{\BL}[1]{#1 \in \mathcal{B}_{\Omega}}
\newcommand{\BLP}[1]{#1 \in \mathcal{B}_{\pi}}
\def\ind{\mathbbmtt{1}}

\makeatletter
\def\moverlay{\mathpalette\mov@rlay}
\def\mov@rlay#1#2{\leavevmode\vtop{%
   \baselineskip\z@skip \lineskiplimit-\maxdimen
   \ialign{\hfil$\m@th#1##$\hfil\cr#2\crcr}}}
\newcommand{\charfusion}[3][\mathord]{
    #1{\ifx#1\mathop\vphantom{#2}\fi
        \mathpalette\mov@rlay{#2\cr#3}
      }
    \ifx#1\mathop\expandafter\displaylimits\fi}
\makeatother

\newcommand{\bigcupdot}{\charfusion[\mathop]{\bigcup}{\cdot}}

\ifCLASSINFOpdf
\else
\fi
\begin{document}

\title{On Unlimited Sampling and Reconstruction}

\author{Ayush~Bhandari,~Felix~Krahmer,~and~Ramesh~Raskar

\thanks{This work was presented in parts at the Intl. Conf. on Sampling Theory and Applications (SampTA), \cite{Bhandari:2017,Bhandari:2019}. A version of this work was included in \cite{Bhandari:2018} and led to the US Patent \cite{Bhandari:2020b}. A.~Bhandari's work is supported by the UK Research and Innovation council's FLF program ``Sensing Beyond Barriers'' (MRC Fellowship award no.~MR/S034897/1). F.~Krahmer acknowledges support by the German Science Foundation (DFG) in the context of the collaborative research center TR 109.}
\thanks{A.~Bhandari was with the Massachusetts Institute of Technology. He is now with the Dept. of Electrical and Electronic Engineering, Imperial College London, South Kensington, London SW7 2AZ, UK. (Email: ayush@alum.mit.edu)}
\thanks{F.~Krahmer is with the Dept. of Mathematics, Technical University of Munich, Boltzmannstra{\ss}e 3, 85748 Garching/Munich, Germany. (Email: felix.krahmer@tum.de)}
\thanks{R.~Raskar is with Media Lab, Massachusetts Institute of Technology, 77 Massachusetts Ave. Cambridge 02139, MA, USA. (Email: raskar@mit.edu)}

}

\markboth{\sf{IEEE Trans. on Sig. Proc. (10.1109/TSP.2020.3041955)}}%
{AB \MakeLowercase{\textit{et al.}}: On Unlimited Sampling}

\maketitle

\begin{abstract} 

Shannon's sampling theorem is one of the cornerstone topics that is well understood and explored, both mathematically and algorithmically. That said, practical realization of this theorem still suffers from a severe bottleneck due to the fundamental assumption that the samples can span an arbitrary range of amplitudes. In practice, the theorem is realized using so-called analog--to--digital converters (ADCs) which clip or saturate whenever the signal amplitude exceeds the maximum recordable ADC voltage thus leading to a significant information loss. 

In this paper, we develop an alternative paradigm for sensing and recovery, called the {\em Unlimited Sampling Framework}. It is based on the observation that when a signal is mapped to an appropriate bounded interval via a modulo operation before entering the ADC, the saturation problem no longer exists, but one rather encounters a different type of information loss due to the modulo operation. Such an alternative setup can be implemented, for example, via so-called folding or self-reset ADCs, as they have been proposed in various contexts in the circuit design literature.

The key task that we need to accomplish in order to cope with this new type of information loss is to recover a bandlimited signal from its modulo samples. In this paper we derive conditions when perfect recovery is possible and complement them with a stable recovery algorithm. The sampling density required to guarantee recovery is independent of the maximum recordable ADC voltage and depends on the signal bandwidth only. Our recovery guarantees extend to measurements affected by bounded noise, which includes the case of round-off quantization. 


Numerical experiments validate our approach. For example, it is possible to recover functions with amplitudes orders of magnitude higher than the ADC's threshold from quantized modulo samples upto the unavoidable quantization error. 


Applications of the unlimited sampling paradigm can be found in a number of fields such as signal processing, communication and imaging.
\end{abstract}

\begin{IEEEkeywords}
Analog-to-digital conversion (ADC), approximation, bandlimited functions, modulo, Shannon sampling theory.
\end{IEEEkeywords}

\newpage

\tableofcontents

\newpage

%
\IEEEpeerreviewmaketitle

\linespread{1.2}
\section{Introduction}

\IEEEPARstart{T}{he} interplay between the continuous and the discrete realms is at the heart of all modern information processing systems. Thus, the sampling theory inevitably finds its way in different spheres of science and engineering. 

Since its early days, the topic has grown far and wide \cite{Butzer:1992,Zayed:1993,Unser:2000}. Researchers and practitioners have proposed a number of extensions. When we think of sampling theory, in most cases, variation on the theme arises from diversity along the \emph{time-dimension}. This could be due to priors (sparsity or smoothness) or sampling architecture (uniform or non-uniform sampling grid). On the other hand, a hypothesis on the {\em amplitude dimension} leads to a different class of sampling problems, namely, level cross sampling, amplitude sampling, quantization and phaseless reconstruction. Our work falls into the latter category. In some sense, it can be seen as an amplitude analogue of the time-warped sampling theory \cite{Cochran:1990}, as in, our work exploits warping along the amplitude axis. 

\subsection{The Sampling Theorem and a Practical Bottleneck}
The key result of sampling theory is Shannon's sampling theorem \cite{Unser:2000}, which states that functions of bounded frequency range can be perfectly recovered from equidistant samples. To realize this theorem in practice, one uses the so-called analog--to--digital converter, henceforth ADC. Unlike the input of the sampling theorem which are assumed to be perfect samples, the output of an ADC is clipped at some fixed maximum recordable voltage; ADCs are limited by their \emph{dynamic range}. This means that only values in a particular range, say, $\left[ { - \lambda ,\lambda } \right]$, may be recorded. We will refer to $\lambda>0$ as the ADC \emph{threshold}. Whenever the input signal amplitude exceeds the threshold, that is $ | \fin| > \lambda$, the ADC saturates and essentially outputs $\fo = \lambda$. This is shown in \fig{fig:adc} (a) and \fig{fig:adc} (b). Such cases cannot be handled by sampling theory. 

Clipping or saturation poses a serious problem in a variety of applications. For example, a camera facing the sun leads to an all white photograph which is the basis of \emph{high dynamic range} photography \cite{Debevec:1997}. This effect is particularly important in the context of self-driving vehicles. Scenarios such as moving out of a tunnel in the day \cite{Yamada:1998} or a headlight in the line-of-sight {can cause the imaging sensors to saturate}. Similarly, in scientific imaging systems such as ultrasound \cite{Olofsson:2005}, radar \cite{Cassidy:2009} and seismic imaging \cite{Zhang:2016}, strong reflections or pulse echoes blind the sensor. In communication systems, clipping results in performance degradation \cite{Li:2015}. In the context of music, clipped sound results in high frequency artifacts \cite{Adler:2012}. 

\begin{figure*}[!t]
\centering
\includegraphics[width =1\textwidth]{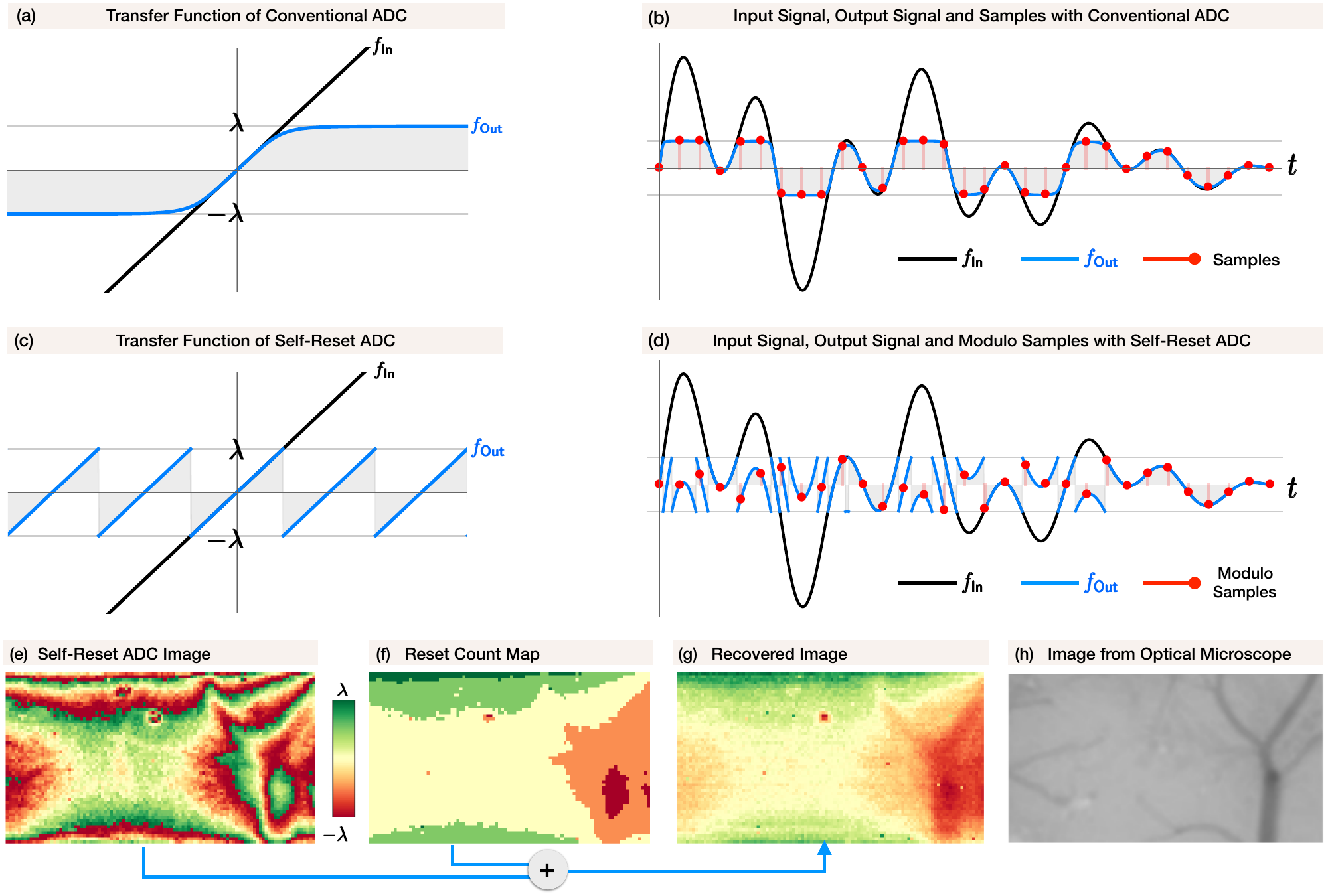}
\caption{Conceptual difference between conventional ADC and the self-reset ADC. (a) Transfer function of conventional ADC. Whenever the maximum amplitude of the input signal $\fin$ exceeds a certain threshold $\lambda$, the output signal $\fo$ saturates to $\lambda$ and this results in clipping. (b) Input signal, output signal and pointwise samples in case of conventional ADC. (c) Transfer function of self-reset ADC. In contrast to conventional ADCs, whenever $| \fin |>\lambda$, the self-reset ADC folds $\fin$ (by using modulo non-linearity) such that $\fo$ is always in the range $\left[-\lambda,\lambda\right]$. In this way, the self-reset configuration circumvents clipping but introduces discontinuities. (d) Input signal, output signal and pointwise modulo samples in case of self-reset ADC. (e) Image obtained with a prototype self-reset ADC shows folded amplitudes \cite{Rhee:2003,Sasagawa:2016,Yamaguchi:2016}. (f) For each pixel, the ``reset count map'' shows the number of times the image amplitude has undergone folding. (g) Unfolded image is obtained by adding reset count map to the folded image. (h) Corresponding optical image (ground truth).}
\label{fig:adc}
 \end{figure*}

Due to the pervasiveness of the clipping effect, several papers have studied this problem in different contexts (cf.~\cite{Abel:1991,Olofsson:2005,Adler:2012,Ting:2013}) and particular emphasis has been put on audio signals (cf.~\cite{Adler:2012} and references therein)---specific examples of bandlimited signals. In general, clipping of a smooth signal results in non-smooth features which in turn corresponds to high-frequency distortions in the signal. Hence, clipped or saturated signals are prone to aliasing artifacts \cite{Esqueda:2016}. 

On one hand, purely computational approaches involve de-clipping or restoration algorithms \cite{Logan:1984,Janssen:1986,Abel:1991,Ting:2013,Esqueda:2016} with the hope of recovering lost samples under certain assumptions. Such approaches suffer from two basic drawbacks: 
\begin{itemize}
\item Recovery guarantees in terms of the sampling density are largely {unexplored}. Ideally, the required sampling rate should be completely governed by the signal bandwidth. 

\item The mapping is discrete to discrete rather than incorporating sampling theory and following the goal of minimizing the error of the reconstructed continuous signal.
\end{itemize}
On the other hand, hardware-only approaches tackle the dynamic range problem at the ADC level \cite{Vanderperren:2006,Smaragdis:2009,Anderson:2009}. This results in sophisticated electronic architecture that is application specific and is agnostic to the benefits offered by computational and algorithmic methods.

\subsection{A Solution via Modular Arithmetic} 
\label{sec:SMA}

Unlike the literature discussed above which relies either entirely on computational approaches or only optimizes the hardware, our work is based on a \emph{co-design} approach; we aim to overcome the dynamic range barrier by repurposing hardware in conjunction with new recovery algorithms. 

The key {innovation} of our approach is that instead of (potentially clipped) pointwise samples of the bandlimited function, we work with folded amplitudes with values in the range $\left[ { - \lambda ,\lambda } \right]$. Mathematically, this folding corresponds to injecting a non-linearity in the sensing process. This amounts to,
\begin{equation}
\label{map}
\mathscr{M}_{\lambda}:f \mapsto 2\lambda \left( {\fe{ {\frac{f}{{2\lambda }} + \frac{1}{2} } } - \frac{1}{2} } \right),
\end{equation}
where $\ft{f} \DE f - \flr{f} $ defines the fractional part of $f$ and $\lambda>0$ is the ADC threshold. Note that (\ref{map}) is equivalent to a centered modulo operation since $\mathscr{M}_{\lambda}(f)\equiv \md{f}{2\lambda}$. By implementing the mapping \eqref{map}, it is clear that out-of-range amplitudes are \emph{folded} back into the dynamic range $\left[ { - \lambda ,\lambda } \right]$.

To connect this mathematical conceptualization to real world applications, we capitalize on recent advances in the imaging and sensor design technology. Indeed, ADC architectures {that are} moving towards the goal of folded sampling are rapidly developing. We have dedicated Section~\ref{sec:Background} to the discussion of such ADCs. In essence, the so-called \emph{self-reset} \cite{Rhee:2003} or \emph{folding} \cite{Kester:2009} ADCs implement folding of amplitudes via \eqref{map} using electronic circuitry. We compare the transfer function of the conventional ADC with the self-reset ADC (henceforth \sr{}) in \fig{fig:adc} (a) and \fig{fig:adc} (c), respectively. Folding the amplitudes ensures that the entire range of the ADC is utilized (cf.~\fig{fig:adc} (c) and \fig{fig:adc} (d)). To give the reader an idea of the functionality of this new breed of ADCs, we show the raw output of the prototype\footnote{We thank Sasagawa \cite{Sasagawa:2016}, Yamaguchi \cite{Yamaguchi:2016} and co-workers for providing the data which allowed us to reproduce these images.} \sr{} arising in the context of imaging in \fig{fig:adc} (e) and \fig{fig:adc} (f).

While the recent decade has seen remarkable progress on the hardware aspects of the new ADCs, theoretical and algorithmic aspects has not been a major research focus and did not make their way to other fields. Current \ab{hardware} literature \cite{Rhee:2003,Sasagawa:2016,Yamaguchi:2016} employs a fairly elementary approach for recovery that uses both the modulo samples and the residuals or reset counts\footnote{As per current approaches, any function (signal or image) can be written as a sum of the folded function and the reset count map. By recording both the folded function as well as the reset count map simultaneously, the original function can be recovered by a straight-forward sum of the two entities. For example, the image in \fig{fig:adc} (g) is obtained by adding \fig{fig:adc} (e) and (f).} (see \fig{fig:adc} (g)). This requires,
\begin{itemize}
  \item complex circuitry and ADC architectures as well as,
  \item additional power and storage.
\end{itemize} 
Furthermore, the feasibility of the reset count for arbitrarily small modulo thresholds has not been investigated yet.

In view of this discussion, we aim for an approach that does not suffer from these shortcomings and recovers the signal without knowledge of residuals or reset-counts. In full generality, this is a very difficult problem, which may also explain the limited progress. Namely, without further assumption on the underlying image or signal, the problem is closely related to the phase unwrapping problem, which is known to be highly ill-posed \cite{Zhao:2015} \aq{and typically relies on inversion of the first-order finite difference (cf.~Section~\ref{sec:RPmod} and Itoh's condition.).} Indeed, in both cases, one seeks to ``unwrap'' a discrete function representation from modulo information. \aq{Consequently, algorithms proposed for phase unwrapping also yield solution strategies for the signal unfolding problem \cite{Rieger:2009,Chen:2014}.}

\aq{As in the case of phase unwrapping, however,  these approaches suffer from strong limitations of the dynamic range (cf.~\fig{fig:Result_1} for a numerical comparison).} For more information on the phase unwrapping problem, we refer to the literature overview in Section~\ref{sec:RPmod}.

An important difference, however, is that for the scenario studied in this paper, one has a considerably larger degree of control of the data entering the sensing pipeline. In particular, it is possible to sense redundant information, which in many other setups has been shown to alleviate the ill-posedness and allow for guaranteed recovery. Examples include sigma-delta quantization of bandlimited functions \cite{Daubechies:2003} and compressed sensing \cite{Candes:2006,Donoho:2006}. In analogy to these works, our goal is to explore redundant representations that allow to overcome the limitations of the conventional viewpoint of phase unwrapping.

\subsection{Contributions and Overview of Results}

The goal of this paper is to pose and study the inverse problem of recovering a continuous-time bandlimited function from its noisy folded samples \emph{without} requiring the knowledge of residuals or reset-counts. Our key contributions are as follows:

\begin{enumerate}[leftmargin=25pt,label= $\textrm{C}_\arabic*)$,itemsep=5pt ]

\item We take a first step towards a sampling theory for modulo samples. To this end, 

\begin{itemize}

\item We establish that a finite energy bandlimited function is identifiable by its modulo samples even when the sampling rate is just slightly above Nyquist.

\item We present a sampling theorem which describes sufficient conditions for sampling and reconstruction of bandlimited functions in the context of folded samples. 

\item We show that this sampling theorem includes stability with respect to bounded noise thus covering scenarios with round-off quantization. 

\end{itemize}

\item On the algorithmic front, our sufficiency condition is complemented by a constructive recovery algorithm that is guaranteed to recover the underlying signal and is stable with respect to noise both in theory and in experiments.

\end{enumerate}
Our main contribution leads to the \emph{Unlimited Sampling Theorem} which is summarized below.

\begin{theorem*}[Unlimited Sampling Theorem] Let $g\l t\r$ be a finite energy, bandlimited signal with maximum frequency $\Omega$ and let ${y[k]},k \in \mathbb{Z}$ in (\ref{map}) be the modulo samples of $g\l t \r$ with sampling rate $\T$. Then a sufficient condition for recovery of $g\l t \r$ from $\{y[k]\}_k$ is that $\T \leq \frac{1}{2\Omega e}$ ({up to additive multiples of $2\lambda$}) where $e$ denotes Euler's constant.%
\end{theorem*}

Remarkably, our theorem requires a sampling rate depending on the bandwidth only, independent of the ratio between ADC threshold, $\lambda$, and signal amplitude. This is why we refer to our method as \emph{unlimited sampling}. Our numerical demonstrations in Section~\ref{sec:NE} clearly corroborate that it is possible to recover signals whose amplitude range significantly exceeds the dynamic range of the ADC under consideration, that is, $\max | \fin| \gg \lambda$.

\subsection{Related Problems in Other Fields and Recent Work}
\label{sec:RPmod}

To put our results into context, we briefly compare and contrast topics where the modulo operation arises naturally. A \ab{disparate} survey of the literature is presented below.

\begin{enumerate}[leftmargin =25pt, itemsep = 7pt, label = $\bullet$]

\item {\bf Phase Unwrapping.} The phase unwrapping problem\footnote{We thank Prof.~Laurent Jacques (UC Louvain) and Prof. Alan Oppenheim (MIT) for bringing the topic of phase unwrapping to our notice and clarifying the intricate differences with respect to sampling theory.} \cite{Itoh:1982,Ghiglia:1998,Choi:2007}, which finds widespread applications, specially in imaging,  arises in a number of applications related to imaging (cf.~\cite{Ghiglia:1998} as well as Chapter 3 in \cite{Bhandari:2018}) where one only has access to a sinusoid of varying phase shift and the relevant information is encoded in these shifts, which are assumed to be varying smoothly. Due to the periodicity of the sinusoid, one can only extract this information modulo $2\pi$, so the goal is to recover a smooth phase function that corresponds to these observations. In this sense, the phase unwrapping problem is closely related to the problem studied in this paper, although it arises in quite a different way; the ambiguity is inherent in the sensing problem whereas we deliberately inject this non-linearity to enhance the reconstruction quality. This computational design aspect becomes important as it allows for further modifications such as the introduction of redundancy. \\

Given this similarity it should not come as a surprise that one of the cornerstone methods for phase unwrapping can be seen as a simplified version of the algorithm proposed in this paper. More precisely, when the max-norm of the first-order finite difference of the samples is bounded by $\lambda$ or, \ab{$\left| {\gamma \left[ {k + 1} \right] - \gamma \left[ k \right]} \right| \leqslant \lambda ,\gamma \left[ k \right] = g\left( {kT} \right)$}, one can recover by inverting the first-order finite difference operator on the modulo sequence. This is an established result and is widely known as {\bf Itoh's condition} \cite{Itoh:1982}. Please refer to \fig{fig:PU} for a schematic explanation of Itoh's condition \ab{and \fig{fig:Result_1} for numerical comparison with our approach}. Due to the \ab{\bf redundancy} via oversampling in our sensing setup, we are able to advance this method both in terms of admissible signal amplitudes and stability.

\item {\bf Digital Communications.} In the work dating back to early 1970s, Tomlinson \cite{Tomlinson:1971} and Harashima \cite{Harashima:1972} describe the use of modulo precoding\footnote{We thank Prof.~Robert Gray (Boston University) who shared several historical facts regarding the role of modulo operations in the context of precoding and offering clarifications with reference to his work on modulo sigma-delta modulation that is mentioned in Section \ref{sec:Background}.} in the context of communication over channels with inter-symbol interference. In recent years, modulo operation has also been used in integer-forcing communication \cite{Zhan:2014}. Although the role of such strategies in connection with sampling theory is unclear, there may be potentially interesting links for specific signal structures and reconstruction algorithms for the case of folded samples.  

\item {\bf Rate Distortion Analysis and Prediction Based Recovery.} As early as in 1979, Ericson \& Ramamoorthy \cite{Ericson:1979} considered quantized modulo samples in the context of a source coding scheme they referred to as \emph{Modulo-PCM}. Here, the authors, propose a heuristic decoder based on the Viterbi algorithm for the case of first-order Gauss-Markov processes. In early 90's, Chou \& Gray \cite{Chou:1992} studied the behaviour of quantization noise in the context of modulo sigma-delta modulation. More recently, Boufounos \cite{Boufounos:2012} developed a recovery scheme for modulo quantized measurements that is based on consistent reconstruction. However, all of these approaches are only feasible in limited dimensions and quite costly when the dimension is large. Our approach, in contrast, is stable and efficient and is easily scalable to a large number of samples in a mathematically guaranteed fashion. The approach in \cite{Boufounos:2012} can be made more practicable with the knowledge of additional side-information \cite{Valsesia:2016} via prediction based recovery. Following the initial ideas underlying this work as summarized in \cite{Bhandari:2017}, a modulo sampling based \ab{simulated} hardware implementation and quantization of bandlimited functions was studied by Ordentlich \etal in \cite{Ordentlich2018}, again using additional side information. The authors considered oversampled representation for both single and multiple channel cases. Later, in \cite{Romanov:2019}, Romanov \& Ordentlich studied recovery for non-quantized samples and showed that sampling faster than Nyquist rate is enough under some mild decay conditions. The side information used in \cite{Ordentlich2018} and \cite{Romanov:2019} consists of a certain number of non-modulo samples based on which linear prediction using Chebyshev polynomials \cite{Vaidyanathan:1987} can be employed to extrapolate the bandlimited function. In contrast to all these works, such {\bf side-information is not required} for the method presented in this paper. This is of particular advantage when only a finite number of measurements is available, for instance, in the case of parametric signals \cite{Bhandari2018,Bhandari2018a} as one can not expect that modulo non-linearity will leave any patch, unaffected.  

\item {\bf Modular Exponentiation.} Consider the problem of computing and storing a number $a^b \bmod c$ where $a,b,c$ are all integers. This problem frequently arises in number theory, computer science and cryptography \cite{Takagi:1992}. The literature in this area focuses on $a^b \bmod c$ without computing $a^b$ as the digital storage required by the former is much smaller than the latter. In this spirit, our work is similar to \emph{modular exponentiation} as the ADC threshold $\lambda$ (in analogy to $c$) controls the amount of information/bits required for representing the dynamic range.  

\item {\bf Contemporary Literature (2018 onwards).} During the completion of this manuscript, our first work \cite{Bhandari:2017} was followed up in several papers. Below, we summarize the key results. 

\begin{enumerate}[leftmargin=25pt,label=---]
\itemsep2pt
  \item Rudresh \etal \cite{Rudresh2018} study a wavelet based algorithm for reconstruction of Lipschitz continuous functions in the context of unlimited sensing framework. 
  \item Cucuringu \& Tyagi \cite{Cucuringu2018} investigated a more general setup based on H\"{o}lder continuous functions. They also presented a denoising approach based on a quadratically constrained quadratic program (QCQP) with non-convex constraints. 
  
  \item Compressed sensing of \emph{discrete-time} sparse signals within the unlimited sampling architecture was investigated by Musa \etal in \cite{Musa2018}. The authors developed a generalized approximate message passing approach to reconstruct discrete signals. Further bounds in context of Gaussian matrices were studied in \cite{Shah2018}. 
  
  \item Unlimited sampling of \emph{continuous-time} sparse signals in canonical and Fourier domain, was discussed in our works \cite{Bhandari2018} and \cite{Bhandari2018a}, respectively.  

\item The idea of \emph{one-bit unlimited sampling} was proposed by Graf \etal in \cite{Graf2019}. This approach recovers high dynamic range signals from signed measurements thus adding to the functionality of the sigma-delta based ADCs. 

\item Given modulo measurements of $n$ i.i.d samples drawn from a Gaussian distribution with unknown covariance matrix, Romanov \& Ordentlich studied recovery methods in \cite{Romanov:2019a}. The case of known covariance matrix was covered in the work of Ordentlich \etal \cite{Ordentlich2018}.

\item In \cite{Ji:2019}, the authors extend the idea of unlimited sensing for signals that live on a graph. To this end, the authors propose
an integer programming based algorithm for recovery from modulo samples.

\item Modulo Radon Transform was presented in \cite{Bhandari:2020} and unlimited sampling based \emph{high dynamic range} computed tomography was presented in \cite{Beckmann:2020}.

\item In recent works, \cite{FernandezMenduina:2020,FernandezMenduina:2020a}, the authors presented a multi-channel unlimited sampling approach in the context of sensor array signal processing. 
\end{enumerate}

\end{enumerate}

\subsection{Organization of this Paper}
The paper is organized as follows. In Section \ref{sec:MS}, we discuss ADC architectures that lead to modulo samples. \ab{Section \ref{sec:IUS} is devoted to the study of uniqueness conditions.} Recovery conditions and the reconstruction algorithm are presented in Section \ref{sec:Rec}. Numerical examples that corroborate our theory are discussed in Section \ref{sec:NE}. We conclude this work with future directions in Section \ref{sec:Conc}.

\begin{figure}[!t]
\centering
\includegraphics[width =0.75\columnwidth]{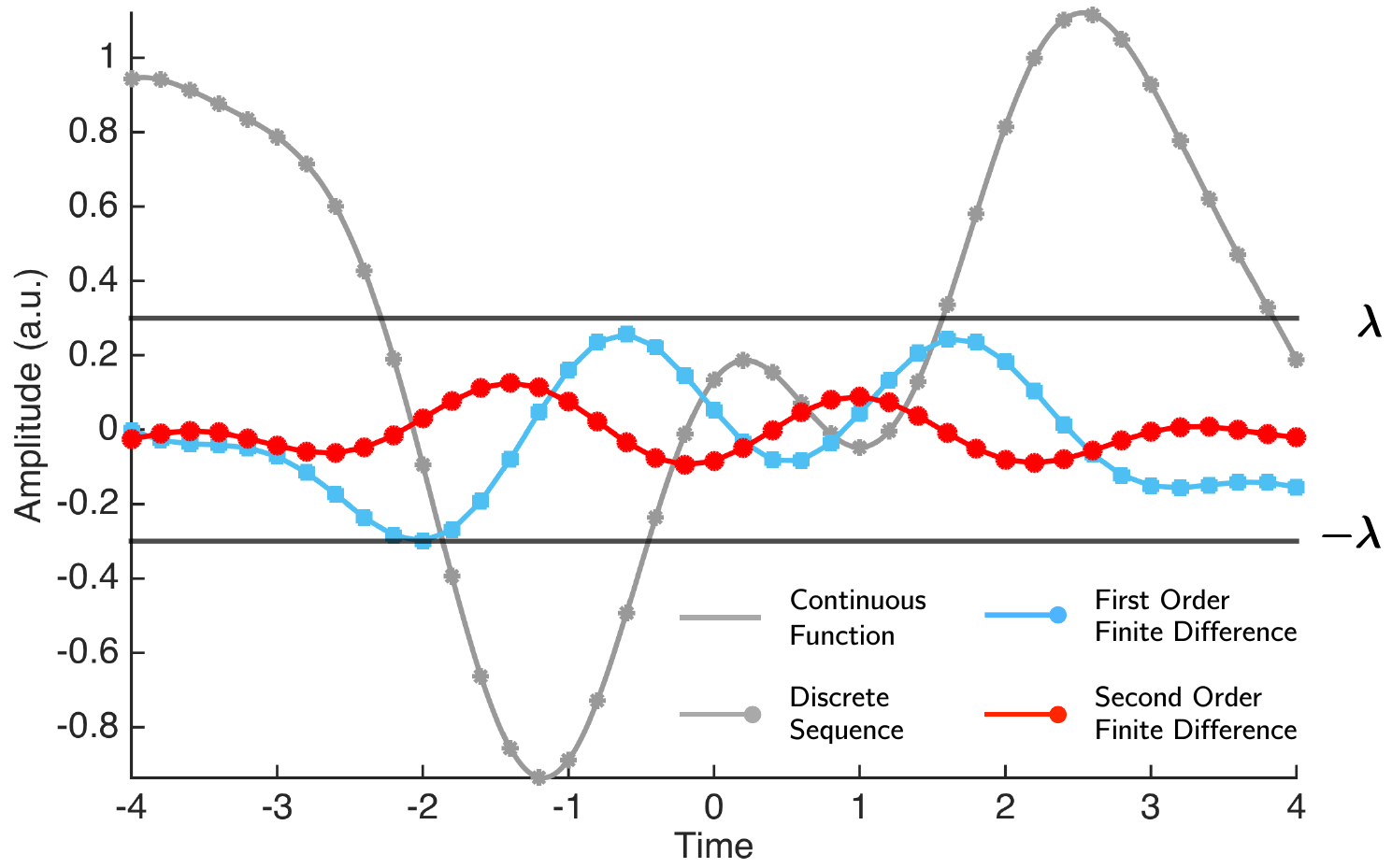}
\caption{Explanation of Itoh's condition \cite{Itoh:1982} in context of phase unwrapping. In the phase unwrapping problem \cite{Ghiglia:1998}, one is given discrete, phase measurements which can only be observed in the range $[0,2\pi]$ and are hence inherently folded. Provided that the max-norm of first-order difference of the samples is bounded by $\lambda=2\pi$, one may recover the original phase by inverting the finite difference. However, this method does not work with higher-order differences due to instabilities and unknown constants involved with the inversion of differences. On the other hand, with increased sampling rate and higher-order differences, one may work with smaller $\lambda$. This is the case with the second-order finite difference (in red). This is the key idea of this paper.}
\label{fig:PU}
 \end{figure}

\subsection{Notation}
We use $\mathbb{N}$, $\mathbb{Z}$, $\mathbb{R}$ and $\mathbb{C}$ to denote the set of natural numbers, integers, reals and complex numbers, respectively. We use $\bigcupdot$ to denote disjoint union of sets. Continuous functions are written as $f\l t \r, t\R$ while their discrete counter-parts are represented as $f\left[k\right],k\Z$. The $L^p$ space equipped with the $p$-norm or $\| \cdot \|_{L^p}$ is the standard Lebesgue space. For instance, $L^1$ and $L^2$ denote the space of absolute and square-integrable functions, respectively. Spaces associated with sequences will be denoted by $\ell^p$. The max-norm ($p\to\infty$) of a function is defined as, ${\left\| f \right\|_{{\infty }}} = \inf \left\{ {{c_0} \geqslant 0:\left| {f\left( t \right)} \right| \leqslant {c_0}} \right\}$ while for sequences, we have, ${\left\| f \right\|_{{\infty }}} = {\max _k}\left| {f\left[ k \right]} \right|$. The $N$--order derivative of a function is denoted by $f^{\l N \r} \l t \r$. The space of $N$--times differentiable, real-valued functions is denoted by $C^N \l \mathbb{R}\r$. For sequences, $$\left( {\Delta f} \right)\left[ k \right] = f\left[ {k + 1} \right] - f\left[ k \right]$$ denotes the finite difference and recursively applying the same yields $N^{\textrm{th}}$ order difference, $\left( {\Delta^N f} \right)\left[ k \right] $. The Fourier transform of a function $f\in L^1$ is denoted by $\widehat{f}\l \omega \r = \int f\l t \r e^{-\jmath \omega t} dt$ with a natural extension to arbitrary distributions via duality. We say $f \in C^N \l \mathbb{R}\r$ is $\Omega$-bandlimited or, 
$$\BL{f} \Leftrightarrow  \widehat f \left( \omega  \right) = {\ind_{\left[ { - {\Omega},{\Omega}} \right]}}\left( \omega  \right)\widehat f \left( \omega  \right)$$ 
where $\ind_{\mathcal{D}}\l t \r$ is the indicator function on domain $\mathcal{D}$. This includes functions with infinite energy such as a pure sine wave. When dealing with finite-energy functions, it is convenient to work with the Paley-Wiener space, $\PW{\Omega} = \mathcal{B}_\Omega\cap {L^2}$. Note that for $\epsilon>0$, one has $\PW{\Omega} \subset \PW{\Omega+\epsilon}$. In the context of numerical computations, $\left\lceil  \cdot  \right\rceil$ and $ \left\lfloor  \cdot  \right\rfloor$ denote the floor and ceiling functions.

\section{\sr{s} and Modulo Samples}

\label{sec:MS}

\subsection{Background on Self-Reset ADCs}

\label{sec:Background}

When reaching the upper or lower saturation voltage, the \sr{s} reset to the respective threshold. This is what allows for capturing voltage variations beyond the conventional saturation limit. This mechanism aptly justifies the name \emph{self-reset} or \emph{folding} ADCs which is mathematically equivalent to using the modulo mapping in \eqref{map}. 

As is often the case, theoretical conceptualization of such ADCs predated its practical implementation. To bring the reader up to pace with the literature, we will quickly review the key references in the area. As early as in the late 1970's, the concept of modulo limiters was used in the context of Tomlinson--Harashima decoders \cite{Tomlinson:1971,Harashima:1972} in communication theory. Chou and Gray \cite{Chou:1992} studied the resulting quantization noise but, neither the physical realization nor their recovery properties were investigated. It was only in the early 2000's that the physical implementation of such ADCs started to develop. One of the earliest references is due to Rhee and Joo \cite{Rhee:2003} where the authors proposed the \sr{} in context of CMOS imagers. Interested readers are also referred to a tutorial article on \emph{folding} ADCs by Kester \cite{Kester:2009} which contains historical references. In a study on quantitative characterization of ADC architectures, Kavusi and Gamal \cite{Kavusi:2004} note that \sr{s} allow for simultaneous enhancement of both dynamic range as well as the signal-to-noise ratio. We remark that the revolution around the \sr{s} is mainly motivated by the goal of achieving improvement on the dynamic range. This is primarily because the dynamic range of natural scenes is typically much larger than what a conventional ADC or an imaging sensor may accommodate. While HDR imaging requires several exposures, \sr{} based imaging is inherently HDR due to folding of amplitudes. An alternative architecture for HDR imaging via modulo operations was recently used in \cite{Zhao:2015}. Beyond consumer photography \cite{Debevec:1997} and autonomous vehicle navigation \cite{Yamada:1998}, improvement on dynamic range is important for scientific and bio-imaging. To this end, Sasagawa \cite{Sasagawa:2016} and Yamaguchi \cite{Yamaguchi:2016} recently developed an implantable \sr{} for functional brain imaging.

\begin{figure}[!t]
\centering
\includegraphics[width =0.8\columnwidth]{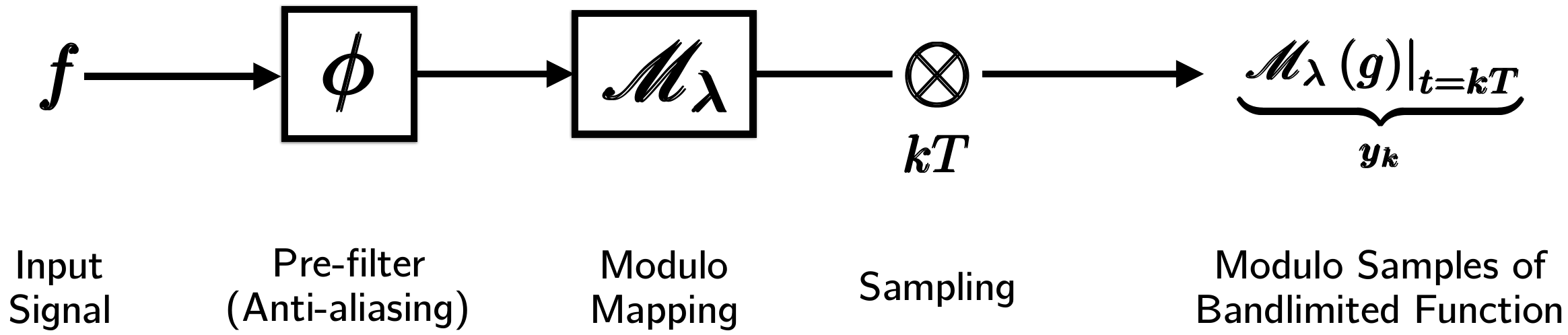}
\caption{Unlimited sampling architecture for obtaining modulo samples.}
\label{fig:blk}
 \end{figure}

\begin{figure*}[!t]
\centering
\includegraphics[width =1\textwidth]{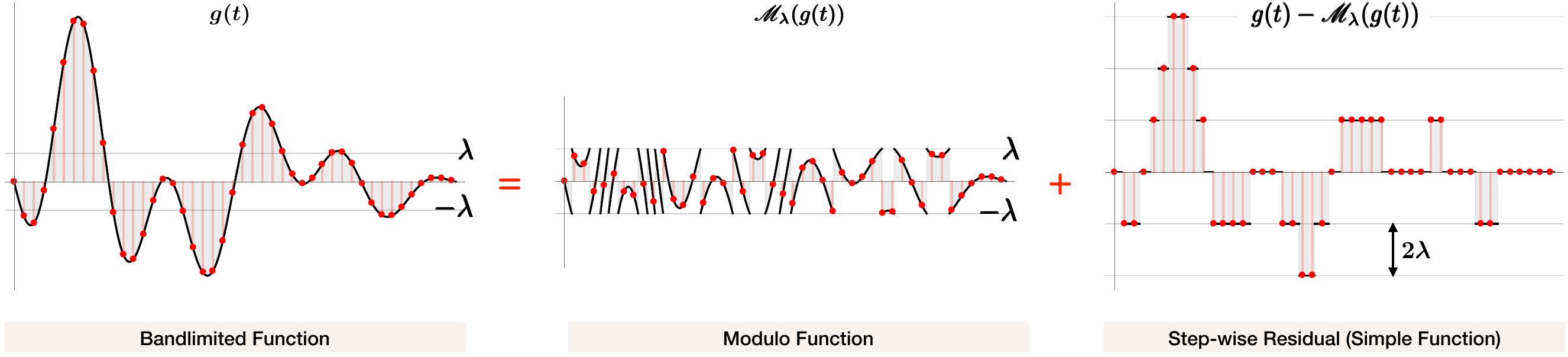}
\caption{Bandlimited functions can be decomposed into a modulo part and corresponding step-wise residual. Same applies to the sampled version (red ink).}
\label{fig:MD}
 \end{figure*}

\subsection{Mathematical Model for Unlimited Sampling}
\label{ssec:MMUS}

\vi{Thanks to the \sr{s} \cite{Rhee:2003,Sasagawa:2016,Yamaguchi:2016}, we can repurpose the resetting capability for obtaining modulo samples without recording the reset counts (cf.~Section \ref{sec:SMA} and \fig{fig:adc} (f)).} The sampling process for obtaining modulo samples of a function is outlined in the block diagram in \fig{fig:blk}. The basic principles are similar to the conventional case except for the modulo mapping. A break down of the key steps is as follows:

\begin{enumerate}[leftmargin=*,label=$\arabic*)$, itemsep = 5pt]

\item We start with a square-integrable function $f\in L^2$ (not-necessarily bandlimited) to be sampled. 

\item Pre-filtering of $f$ with $\phi \in \PW{\Omega} $ results in a low-pass approximation which is given by, 
\begin{align}
\label{gt}
g \in \PW{\Omega} , \quad g\left( t \right) & = \left( {f*\phi } \right)\left( t \right) = \int {f\left( \tau  \right)\phi \left( {t - \tau } \right)dt}.  \notag
\end{align}

Also note that since $g \in \PW{\Omega} $, it admits the standard sampling formula for bandlimited functions, 
\begin{equation}
\label{ST}
0 < \T \leqslant \frac{\pi }{\Omega }, \ \ g\left( t \right) = \sum\limits_{k =  - \infty }^{k =  + \infty } {g\left( {k\T} \right)\operatorname{sinc} \left( {\frac{t}{\T} - k} \right)}, 
\end{equation}
where $\T$ is the sampling rate and $\operatorname{sinc} \left( t \right) = \tfrac{{\sin \left( {\pi t} \right)}}{{\left( {\pi t} \right)}}$ is the ideal low-pass filter.

\item The bandlimited function $g$ is folded in the range $\left[ { - \lambda ,\lambda } \right]$ via non-linear mapping \eqref{map} and results in, 
\begin{equation}
\label{yc}
z\l t \r \DE \MO{g\l t \r}.
\end{equation}

\item The folded function $z\l t \r$ is sampled using impulse-train, {${ \otimes _{k \T }} \DE \sum\limits_{n \in \mathbb{Z}} {\delta \left( {t - k \T }\right)}$}, with sampling rate $\T>0$ yielding uniform samples, 
\begin{equation}
\label{yn}
{y[k]} \DE z\left( {k \T } \right)  = \MO{g\l k\T \r}, \ \ k\Z
\end{equation}
as shown in \fig{fig:blk}.

\item When considering quantization with a budget of $B$ bits per sample, each modulo sample $y[k]$ is rounded to the closest element  in the set 
\[
\EuScript{C}_{B,\lambda} = \left\{ {\left. { \pm \frac{{\left( {2n + 1} \right)}}{{{2^B}}}\lambda } \ \right|n \in \left\{ {0, \ldots ,{2^{B - 1}} - 1} \right\}} \right\}.
\]
Note that this operation induces noise that is not following a random model, but is entry-wise bounded by $2^{-B}$, which is why our guarantee below considers bounded deterministic noise.
\end{enumerate}

\noindent We plot $g\l t \r$, $z\l t \r $ and modulo samples ${\left\{ {{y[k]}} \right\}_{k \in \mathbb{Z}}}$ in \fig{fig:adc}. 

\subsection{The Structure of Discontinuities in Modulo Representation}
Clearly, the unlimited sampling architecture converts a smooth function into a discontinuous one. Recovering the unfolded function $g\l t \r$ from scrambled, low dynamic-range, samples then boils down to patching the discontinuities together. In fact, it turns out that the discontinuities admit a structure which is critical to our cause. Every bandlimited function, continuous or discrete, can be decomposed as a sum of a modulo function and a step-wise residual that we call \emph{simple function}. This is shown in \fig{fig:MD}. We elaborate on this aspect in form of the following proposition. 

\begin{proposition}[Modular Decomposition Property]
\label{prop:mod} Let $\BL{g}$ and $\MO{\cdot}$ be defined in (\ref{map}) where $\lambda$ is a fixed, positive constant. Then, the bandlimited function $g\l t \r$ admits a decomposition
\begin{equation}
\label{eg}
g\l t \r =  z\l t \r +{{\varepsilon _g}\left( t \right)} 
\end{equation}
where $z\l t \r = \MO{g\l t \r}$ and ${\varepsilon _g}$ is a simple function, 
\begin{equation}
\label{eq:SF}
{\varepsilon _g}\left( t \right) = 2\lambda\sum\limits_{m \in {\mathbb{Z}}} {{e[m]}{{\ind}_{{\mathcal{D}_m}}}\left( t \right)}, \ \  e[m] \in \mathbb{Z},  
\end{equation}
where $\bigcupdot_{m\Z} \mathcal{D}_{m} = \mathbb{R}$ is a partition of the real line into intervals $\mathcal{D}_m$.
\label{prp:MD}
\end{proposition}
\begin{proof}
Since $z\l t \r  = \MO{g\l t \r}$, by definition, we write, 
\begin{align}
 {\varepsilon _g}\left( t \right)   \EQc{eg}   g\l t \r - \MO{g\l t \r} & \DL{a} 2\lambda \left(h\l t \r- \ft{h\l t \r} \right) \notag \\
& \DL{b} 2\lambda \flr{h\l t \r}, 
\label{egh}
\end{align}
where (a) is due to $h\DE\l g/2\lambda\r + 1/2 $ and in (b), we use $h = \ft{h} + \flr{h}$. Since, for an arbitrary function $h$, $\flr{h}$ has the form, 
\begin{equation}
\label{sf}
\flr{h\l t \r} = \sum\limits_{\ell \in \mathbb{Z}} {{e[\ell]}{{\ind}_{{\mathcal{D}_\ell}}}\left(t \right)}, \quad e[\ell] \in \mathbb{Z}, \ \ \bigcupdot\limits_{\ell\Z} \mathcal{D}_{\ell} = \mathbb{R}
\end{equation}
{we obtain} the desired result. 
\end{proof}

This proposition is the stepping stone towards the recovery algorithm. It allows us to tackle two problems at once: 
\begin{enumerate}[leftmargin=25pt,label= $  \arabic*)$,itemsep=5pt]
  \item According to \eqref{eg}, if $\varepsilon_g$ is known, we can recover $g$ from $z$. In this paper we develop a method which allows for inferring $\varepsilon_g$ from $z$.  
  \item According to \eqref{sf}, the fact that the amplitudes of $\VO{g}$ can only be integer multiples of $2\lambda$ allows us to enforce a \emph{consistency} constraint\footnote{By the consistency constraint, we mean that in presence of noise or rounding errors, we may force the amplitudes of $\varepsilon_g$ to be on the grid of $2\lambda\mathbb{Z}$. This is implemented using \eqref{eq:RD} and step $4(b)$ in Algorithm \ref{alg:ModSamp}. For numerical demonstration see Section \ref{sec:VS}.} which in turn leads to a robust recovery algorithm.
\end{enumerate}
We emphasize that the decomposition property of Proposition \ref{prop:mod} offers a general purpose utility whose potential benefits are typically not exploited. For instance, phase unwrapping \cite{Itoh:1982} may directly benefit from this observation. In phase unwrapping literature, one recovers $g(nT)$ from its first-order finite difference using the anti-difference. Clearly, an arbitrary smooth function has far more degrees of freedom than \eqref{sf} and hence, from robustness perspective, it is more effective to enforce the simple function structure.

\begin{figure*}[!t]
\centering
\includegraphics[width =1\textwidth]{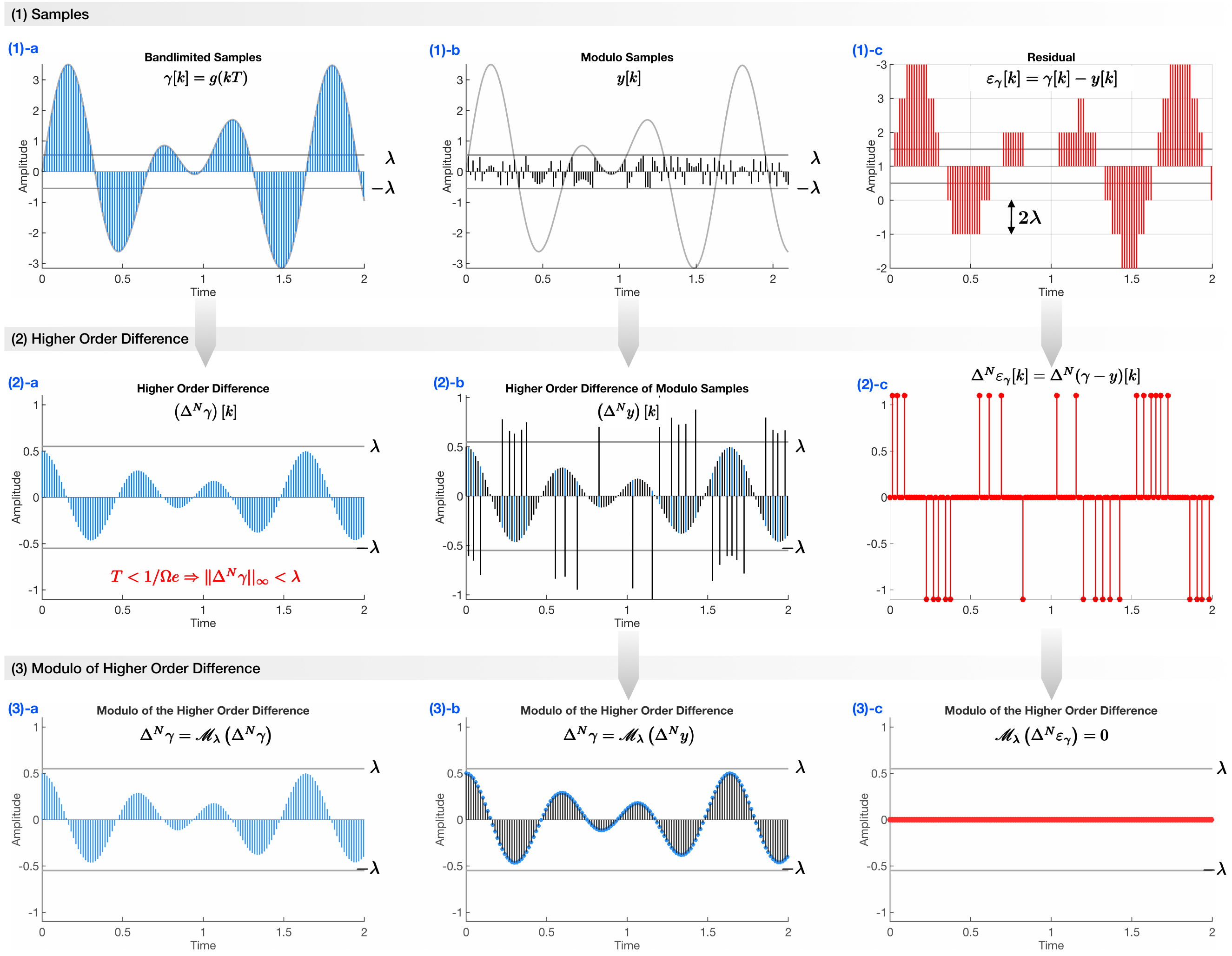}
\caption{Overview of the main idea behind recovering a bandlimited function from modulo samples. Given the sequence of modulo samples, our basic strategy will be to apply a higher order finite difference operators. We will be exploiting that such operators commute with the modulo operation. So after applying the amplitude folding to the resulting sequence, one obtains the same output as if one had started with $\gamma[k] = g(k\T)$ instead of $y[k]$. That in turn will allow for recovery if the higher order finite differences of the $\gamma[k]$'s are so small that the amplitude folding has no effect. Sub-figure (1) shows modulo decomposition, specified by \eqref{eg}, or $\gamma[k] = y[k] + \VO{\gamma}[k]$ . Sub-figure (2) shows pointwise difference of (1)-a, the bandlimited function, pointwise difference of (1)-b, that is, folded version of bandlimited function in (1)-a and pointwise difference of the residual in (1)-c. Sub-figure (3) shows the modulo of plots in (2)-a, (2)-b and (2)-c, respectively. When the sampling density meets certain criterion, the differences are always bounded by the threshold $\lambda$. More precisely, $\T<1/\Omega e \Rightarrow\|{\Delta ^N}\gamma |{|_\infty } < \lambda$ as shown in Section \ref{subsec:RHODMS}. Hence modulo operation has no effect implying that plots (2)-a and (3)-a exactly coincide. As shown in Sub-figure (3)-b, the plots in (3)-a and (3)-b exactly coincide (also see \eqref{eq:modeq}). Hence, starting with the modulo measurements $y[k]$ in (1)-b, it is possible to recover the higher differences of (1)-a, that is, $\ND{\gamma}$ which is shown in (2)-a. To recover $\gamma[k]$ from $\ND{\gamma}[k]$, we first estimate $\ND{\VO{\gamma}} = \MO{\ND{y}} - \ND{y} $. Since $\VO{\gamma}$ only takes amplitudes on the grid of $2\lambda\mathbb{Z}$, we exploit this restriction on amplitudes to recover $\VO{\gamma}$ in (1)-c. There on, adding $\VO{\gamma}[k]$ in (1)-c and the modulo samples $y[k]$ in (1)-b results in the bandlimited samples in (1)-a.}
\label{fig:Main}
 \end{figure*}

\section{On Identifiability in Unlimited Sampling }

\label{sec:IUS}

In this section, we \ab{present} identifiability conditions associated with the unlimited sampling strategy. Said differently, we answer the question of what sampling density leads to unique characterization of a bandlimited function in terms of its modulo samples? The key result of this section is that any sampling rate faster than critical sampling, that is, $0<\T<\pi/\Omega$ allows for a one-to-one mapping between a bandlimited function, $f\in \PW{\Omega}$ and modulo samples ${y[k]}  = \MO{g\l k\T \r}$. In what follows, without loss of generality, we will normalize the bandwidth such that $\Omega = \pi$ and the critical sampling rate is $\T = 1$. In working with oversampled representations, we will denote the sampling rate by, $\TE=\frac{\pi}{\pi+\epsilon}, \epsilon>0$ implying that $0<\TE<1$.

\subsection{Uniqueness of Modulo Representation}

The following Lemma is used for the proof of injectivity conditions.  

\begin{lemma} \label{lem} Assume that $f\in\PW{\pi} $, $\epsilon>0$ and $\mathbb{L} \subset \mathbb{Z}$ is some finite set. Then $f$ is uniquely characterized by its samples in $ \TE\cdot(\mathbb{Z}\setminus \mathbb{L})$ where $\TE=\frac{\pi}{\pi+\epsilon}, \epsilon>0$ is the sampling rate.
\end{lemma}

\ab{For proof of this Lemma, we refer the reader to \cite{Bhandari:2019}.} Lemma~\ref{lem} is used to prove that oversampling uniquely determines modulo samples. \ab{The formal statement is as follows.}


\begin{theorem}[Injectivity Theorem for Unlimited Sampling]\label{thm:IUS} Any $f\in \PW{\pi}$ is uniquely determined by its modulo samples on the grid ${\left\{ {{t_n} = n{\TE }} \right\}_{n \in \mathbb{Z}}}$ with $\epsilon>0$.
\end{theorem}
%
\ab{For proof of this Theorem, we refer the reader to \cite{Bhandari:2019}. A subsequent alternative form of proof can be found in \cite{Romanov:2019}.} 

\section{Recovery Conditions and a Reconstruction Algorithm}
\label{sec:Rec}
The result of the previous section shows that any sampling rate faster than the Nyquist rate uniquely characterizes a bandlimited function in terms of its modulo samples. However, this does not yield a constructive algorithm. This section is concerned with the development of a constructive algorithm that recovers a bandlimited function from its modulo samples as defined in \eqref{yn}. Our algorithm is accompanied with a recovery guarantee. Akin to Shannon's sampling theorem, our recovery guarantee purely depends on the signal bandwidth.

\subsection{Overview of the Recovery Scheme}	

Our basic strategy for recovering functions from modulo samples is schematically explained in \fig{fig:Main}. The key observation at the heart of our approach is that finite differences commute with the modulo operation in a certain sense. With the first order difference given by, $\left( {\Delta y} \right)\left[ k \right] \DE y\left[ {k + 1} \right] - y\left[ k \right]$, the $N^{\textrm{th}}$ order difference can be obtained by recursive application of the finite-difference operator, ${\Delta ^N}y= {\Delta ^{N - 1}}\left( {\Delta y}\right)$. Starting with modulo samples $y[k]$ (cf. \fig{fig:Main}-(1)-b), its higher order differences are given by $\left( {{\Delta ^N}y} \right)\left[ k \right]$ (cf. \fig{fig:Main}-(2)-b). Under appropriate conditions, as shown in \fig{fig:Main}-(3)-b, applying the modulo operation on $\left( {{\Delta ^N}y} \right)\left[ k \right]$ reduces to $\left( {{\Delta ^N}g} \right)\left[ k \right]$ and we obtain the equivalence $\MO{ {{\Delta ^N}y}}\left[ k \right] = \left( {{\Delta ^N}g} \right)\left[ k \right]$. Equivalently, the residual function is annihilated that is, $\MO{\Delta^N \varepsilon_g} = 0$. To eventually recover $g\left[ k \right]$ from $\MO{ {{\Delta ^N}y}}\left[ k \right] = \left( {{\Delta ^N}g} \right)\left[ k \right]$, we will work towards recovery of $\varepsilon_g$ which can be robustly estimated as it takes amplitudes on the grid of $2\lambda\mathbb{Z}$. To do so, we compute, $\left( {{\Delta ^N}y} \right) - \MO{\Delta ^N y} = \Delta ^N\varepsilon _g.$ For $N=1$, this is shown in \fig{fig:Main}-(2)-c. To recover $\varepsilon _g$ from $\DN{N}{\VO{g}}$, we will use the anti-difference operator defined as, 
\begin{equation}
\label{eq:AD}
\S:{\left\{ {s\left[ k \right]} \right\}_{k \in {\mathbb{Z}^ + }}} \to \sum\limits_{m = 1}^k \ab{s\left[ m \right]}. 
\end{equation}
followed by rounding of $\S\l \DN{N}{\VO{g}}\r$ to the nearest multiple of $2\lambda\mathbb{Z}$ which adds to the stability. Although polynomials are in the kernel of the anti-difference operator, we will develop an approach that allows us to estimate $\varepsilon _g$ up to an unknown constant. Finally, having estimated the residual function $\widetilde\varepsilon _g$ up to an integer multiple of $2\lambda$, we obtain the bandlimited samples $g\left[ k \right]$ by using \eqref{eg}. The continuous-time function is obtained by low-pass filtering the samples using \eqref{ST}.

\subsection{Towards a Recovery Guarantee for Unlimited Sampling}
Our scheme relies on conditions under which the higher order differences of the bandlimited samples are small enough in amplitude so that the modulo non-linearity has no effect on the same. To analyze how much can one shrink the amplitudes, we will begin our analysis with a bound that relates max-norm of higher order differences with its continuous-time counterpart, the derivative. Below, we summarize some well-known consequences of Taylor's theorem and Bern\v{s}te\u{\i}n's inequality in form of the following lemma. 

\begin{lemma}
\label{Lemma}
For any $g\in C^m(\mathbb{R}) \cap L^\infty(\mathbb{R}) $ and $N\in\mathbb{N}$, its samples $\gamma [k]\DE g(k\T)$ satisfy 
\begin{equation}
\label{lem:1}
\|\Delta^N \gamma\|_\infty \leq {\left( {\T e} \right)^N} \| g^{(N)}\|_\infty.
\end{equation}
Furthermore, whenever $g\l t \r$ is an $\Omega$--bandlimited function, its samples $\gamma [k]\DE g\l k\T\r, \T\leq\pi/\Omega$ satisfy, 
\begin{equation}
\label{lem:1b}
\|\Delta^N \gamma\|_\infty \leq {\left( {\T \Omega e} \right)^N} \| g\|_\infty.
\end{equation}
\end{lemma}

\begin{proof} Using Taylor's theorem, we can write, 
\begin{equation}
\label{eq:TSE}
g\left( t \right) = {P_{N - 1}}\left( t \right) + {R_{N - 1}}\left( t \right)
\end{equation}
where ${P_{N-1}}\left( t \right)$ is the Taylor polynomial of degree $N-1$ around $\tau$, 
$${P_{N - 1}}\left( t \right) = \sum\limits_{n = 0}^{N - 1} {\frac{{{g^{\left( {n - 1} \right)}}\left( \tau \right)}}{{\left( {n - 1} \right)!}}{{\left( {t - \tau } \right)}^{\left( {n - 1} \right)}}}$$
and the remainder of the Taylor series in \emph{Lagrange form} is given by, $${R_{N - 1}}\left( t \right) = \frac{{{g^{\left( N \right)}}\left( z \right)}}{{N!}}{\left( {t - \tau } \right)^N}$$ where $z$ is some number between $\tau$ and $t$. Given a sequence of samples $\gamma[k] = g\l k\T\r$, its higher order differences $\left( {{\Delta ^N}\gamma } \right)\left[ k \right]$ use $N$ contiguous values of $g\l t_n \r$ where, $$t_n = \l k+n\r \T \in \left[ {k \T,\left( {n + k} \right)\T} \right], \quad n = 0, \ldots ,N.$$
By letting $\tau$ to be the center of the above and hence $\tau = \l k+\frac{N}{2}\r \T$ and sampling \eqref{eq:TSE} on both sides, we can now write $$\gamma \left[ k \right] = {P_{N - 1}}\left( k\T \right) + {R_{N - 1}}\left( k\T \right)$$ and consequently,
\begin{equation}
\label{eq:b1}
\left( {{\Delta ^N}\gamma } \right)\left[ k \right] = \left( {{\Delta ^N}{R_{N - 1}}} \right)\left( k \T \right)
\end{equation}
because $\Delta^N$ annihilates polynomials of degree $N-1$. Noting that ${\left\| {{\Delta ^N}{R_N}} \right\|_\infty } \leqslant {2^N}{\left\| {{R_N}} \right\|_\infty }$ and using \eqref{eq:b1}, we have, 
\begin{align*}
  {\left\| {{\Delta ^N}\gamma } \right\|_{{\infty }}}   & \leqslant 
 {{2^N}\frac{{{{\left\| {{g^{(N)}}} \right\|}_\infty }}}{{N!}}{{\left( {\frac{N}{2}{\T}} \right)}^N}}   \\
& \leqslant \frac{{{{\left( {N\T} \right)}^N}}}{{N!}}\|{g^{\left( N \right)}}|{|_\infty }\\  
 &\leqslant \frac{{{{\left( {\T e} \right)}^N}}}{{\sqrt {2\pi N} }}\|{g^{\left( N \right)}}|{|_\infty } \hfill   \leqslant \left( {\T e} \right)^N\|{g^{\left( N \right)}}|{|_\infty }
\end{align*}
where the last but one inequality is obtained by applying the Stirling's approximation, that is, $N!\sim\sqrt {2\pi N} {\left( {\frac{N}{e}} \right)^N}$ and this proves \eqref{lem:1}. 

Next, we use the well known \emph{Bern\v{s}te\u{\i}n's inequality} to bound the right hand side of \eqref{lem:1}. For $\Omega$--bandlimited functions $g\in\PW{\Omega}$, the Bern\v{s}te\u{\i}n's inequality (cf.\ pg.\ 116, \cite{Nikolskii:1975}) asserts that, 
\[
\|{g^{\left( N \right)}}|{|_p } \leqslant {\Omega ^N}\|g|{|_p }, \quad 1\leq p \leq \infty
\]
and by plugging the same in \eqref{lem:1}, we obtain \eqref{lem:1b}.
\end{proof}

\subsubsection{Recovering Higher Order Differences from Modulo Samples}

\label{subsec:RHODMS}

The result of Lemma \ref{Lemma} and in particular, the inequality in \eqref{lem:1b} will be the key to our proposed recovery method. By letting $\T < 1/\Omega e$ and choosing $N$ logarithmically in $\|g\|_\infty$, the right hand side of \eqref{lem:1b} can be made to shrink arbitrarily, ultimately ensuring that $\|{\Delta ^N}\gamma |{|_\infty } < \lambda$. More precisely, assuming that we know some constant $\B$ such that, 
\[
\B \in 2\lambda\mathbb{Z} \quad \mbox{ and } \quad \|g\|_\infty \leq \B
\]
a suitable choice for $N$ that satisfies ${\left( {\T\Omega e} \right)^N}\B \leq \lambda$ is,
\begin{equation}
\label{eq:NO}
\quad N \geq \left\lceil {\frac{{\log \lambda  - \log \B}}{{\log \left( {\T\Omega e} \right)}}} \right\rceil.
\end{equation}
For this choice of $N$ and $\T \leq {\left( {2\Omega e} \right)^{ - 1}}$, \eqref{lem:1b} entails that $ \left( {{\Delta ^N}\gamma } \right)\left[ k \right]$ is unaffected by the modulo operation, that is,
\begin{equation}
\label{eq:modeq}
{\Delta ^N}\gamma \equiv \MO{{\Delta ^N}\gamma}.
\end{equation}
The following proposition relates the right hand side of \eqref{eq:modeq} to the modulo samples $y[k]$ in \eqref{yn}.

\begin{proposition} 
\label{prp:COM}
For any sequence $a[k]$ it holds that
\begin{equation}
\label{yngn}
{\mathscr{M}_{\lambda}(\Delta^N a) =\mathscr{M}_{\lambda}(\Delta^N({\mathscr{M}_{\lambda}(a)}).} 
\end{equation}
\end{proposition}

\begin{proof}
In view of the \emph{Modular Decomposition Property}, we write, $a = \MO{a} + \VO{a}$ to obtain $\Delta^N \MO{a}= \Delta^N {a} - \Delta^N\VO{a}$ where $\VO{a}$ is a simple function. With the observation that $$\MO{a_1 + a_2} = \MO{\MO{a_1} + \MO{a_2}},$$ we obtain, 
\begin{align*}
\MO{\Delta^N{\MO{a}}} & = \MO{  \Delta^N a -  \Delta^N \VO{a}} \\ 
&= \MO{\MO{\Delta^N a} - \MO{\Delta^N \VO{a}} }. 
\end{align*}
Now since $\VO{a}$ in \eqref{eq:SF} take values in $2\lambda\mathbb{Z}$ and the finite difference filter has integer valued coefficients, it follows that $\Delta^N\VO{a} \in 2\lambda\mathbb{Z}$. Consequently, $\Delta^N\VO{a}$ is in the kernel of $\MO{\cdot}$ or $\MO{\Delta^N\VO{a}} = 0$. Hence, 
$$\MO{\Delta^N{\MO{a}}} = \MO{\MO{\Delta^N{a}}}= \MO{\Delta^N{a}},$$
which proves the result in \eqref{yngn}.
\end{proof}

By letting $a = \gamma$ in \eqref{yngn}, we obtain, $\MO{\Delta^N \gamma} = \MO{\Delta^N y}$. Combining this with \eqref{eq:modeq} yields, 
\begin{equation}
\label{eq:modeq2}
\Delta^N  \gamma = \MO{\Delta^N y}.
\end{equation}
Hence, starting with modulo samples $y$ in \eqref{yn}, we are able to relate the same with higher order differences of bandlimited samples $\gamma$.

\subsubsection{Recovery of Samples from Higher Order Differences} 
\label{sec:RSHOD}
The result of Proposition \ref{prop:mod} yields that, 
\begin{equation}
\label{eq:epsgamma}
\VO{\gamma}[k] \DE \gamma[k] -y[k] =: 2\lambda\sum\limits_{m \in {\mathbb{Z}}} {{e[m]}{{\ind}_{{\mathcal{D}_m}}}\left( k\T \right)}, \quad 
\end{equation}
where the sequence $e$ satisfies $e[m] \in \mathbb{Z}$ for all $m\in \mathbb{Z}$.
Hence given modulo samples $y[k]$, recovering $\gamma[k]=g\l k\T \r$ is equivalent to recovering $\VO{\gamma}[k]$. Thanks to Proposition \ref{prp:COM}, with $N$ in \eqref{eq:NO}, we can use the relation in \eqref{eq:modeq} to estimate $\ND{\VO{\gamma}}$ directly from the modulo samples. This is because,
\begin{equation}
\label{eq:egy}
\ND{\VO{\gamma}} = \ND{\l \gamma - y \r} = \MO{\DN{N}{y}} - \DN{N}{y}. 
\end{equation}
It now remains to recover $\VO{\gamma}$ from $\VOB{\gamma}$. We remind the reader that $\VOB{\gamma}\in2\lambda\mathbb{Z}$ and this massive restriction on the range of values that $\VOB{\gamma}$ can take makes the recovery procedure considerably less ill-posed than recovering $\gamma$ from $\ND{\gamma}$. Algorithmically, recursively applying the anti-difference operator $\S$ in \eqref{eq:AD} and then rounding the same to the nearest multiple of $2\lambda$, that is, 
\begin{equation}
\label{eq:RD}
\S \VOB{\gamma} \mapsto 2\lambda \Big\lceil {\frac{{\left\lfloor {\S \VOB{\gamma}/\lambda } \right\rfloor }}{2}} \Big\rceil
\end{equation}
recovers $\DN{N-1}{\VO{\gamma}}$ up to an additive constant. The remaining ambiguity is due to the constant sequence being in the kernel of the first order finite difference operator. The unknown constant can again only take values in $2\lambda\mathbb{Z}$. More precisely, 
\begin{equation}
\label{eq:S1}
\l \Delta^{n-1} \VO{\gamma}\r [k] = \l \S \Delta^{n} \VO{\gamma} \r [k] + \kappa_n l[k]
\end{equation}
where $ l[k] = 2\lambda\mathbb{Z}$ is the constant sequence and $\kappa_n \Z$. Clearly, when $n=1$, the ambiguity cannot be resolved as $\gamma[k] + 2\lambda\mathbb{Z}$ leads to the same modulo samples. However, for ${1<n\leq N}$, we can resolve this ambiguity. By applying $\S$ to \eqref{eq:S1}, we obtain,
\begin{equation}
\label{eq:S2}
\l \Delta^{n-2} \VO{\gamma}\r [k] = \l \S^2 \Delta^{n} \VO{\gamma} \r [k] +  \kappa_{n} \S l[k] +  \kappa_{n-1} l[k]
\end{equation}
where $\l\S l\r[k] ={{{\left\{ {2\lambda k} \right\}}_{k \in {\mathbb{Z} }}}}$ is a linear sequence with slope $2\lambda$. Next, observe that, 
\begin{equation}
\begin{split} &\left( {{\S^2}{\Delta ^n}{\varepsilon _\gamma }} \right)\left[ 1 \right]  - \left( {{\S^2}{\Delta ^n}{\varepsilon _\gamma }} \right)\left[ {J + 1} \right] \hfill \\&   \EQc{eq:S2} \left( {{\Delta ^{n - 2}}{\varepsilon _\gamma }} \right)\left[ 1 \right] - \left( {{\Delta ^{n - 2}}{\varepsilon _\gamma }} \right)\left[ {J + 1} \right] + 2\lambda {\kappa _n}J \hfill \\&   \DA{a} \left( {{\Delta ^{n - 2}}\gamma } \right)\left[ 1 \right] - \left( {{\Delta ^{n - 2}}y} \right)\left[ 1 \right] - \left( {{\Delta ^{n - 2}}\gamma } \right)\left[ {J + 1} \right] \\ 
& \qquad + \left( {{\Delta ^{n - 2}}y} \right)\left[ {J + 1} \right] + 2\lambda {\kappa _n}J
\end{split}
\label{S2}
\end{equation}
where $(\mathrm{a})$ uses the modulo decomposition property \eqref{eg}. Furthermore, $\|y\|_\infty \leqslant \lambda$ implies that $\left( {{\Delta ^{n - 2}}y} \right) \leqslant {2^{n - 2}}\lambda$ and from Lemma \ref{Lemma}, \eqref{lem:1b}, we have $\|{\Delta ^{n - 2}}\gamma {\|_\infty } \leqslant {\left( {\T\Omega e} \right)^{n - 2}}\B$. Hence, 
\begin{equation*}
\begin{split}
& |  \left( {{\Delta ^{n - 2}}\gamma } \right)\left[ 1 \right]  - \left( {{\Delta ^{n - 2}}\gamma } \right)\left[ {J + 1} \right]  -\left( {\left( {{\Delta ^{n - 2}}y} \right)\left[ 1 \right] - \left( {{\Delta ^{n - 2}}y} \right)\left[ {J + 1} \right]} \right)| \\
&  \leq 2{\left( {\T\Omega e} \right)^{n - 2}}\B + 2^{n-1}\lambda.
\end{split}
\end{equation*}
From the above inequality, it is clear that oversampling or $\T < 1/\Omega e$ keeps the max-norm of the higher order differences bounded. Using this with \eqref{S2}, we conclude,
\begin{equation}
\label{eq:interbound}
\begin{split}
&\left( {{\S^2}{\Delta ^n}{\varepsilon _\gamma }} \right)\left[ 1 \right]  - \left( {{\S^2}{\Delta ^n}{\varepsilon _\gamma }} \right)\left[ {J + 1} \right] \\ 
& \in 2\lambda {\kappa _n}J + \l2{\left( {\T\Omega e} \right)^{n - 2}}\B + 2^{n-1}\lambda\r \left[-1,1 \right]\\
& \subseteq 2\lambda J\left[ {{\kappa _n} - {\rho _n}\left( {\lambda ,\B} \right),{\kappa _n} + {\rho _n}\left( {\lambda ,\B} \right)} \right]
\end{split}
\end{equation}
where,
\begin{equation}
\label{eq:RN}
{\rho _n}\left( {\lambda ,\B} \right) \DE \frac{1}{J}\left( {\frac{\B}{\lambda } + {2^{n - 2}}} \right).
\end{equation}
Here, the last inclusion follows as $\T\leq \l 2\Omega e\r^{-1}$. Using $n\leq N$ and $\T\leq \l 2\Omega e\r^{-1}$, again, we have
\[
{2^{n - 1}} \leqslant {2^{N - 1}} \leqslant {\left( {{\B}/{\lambda }} \right)^{ - \frac{1}{{{{\log }_2}\left( {\T\Omega e} \right)}}}} \leqslant {\B}/{\lambda }.
\]
With \eqref{eq:RN}, the above inequality yields ${\rho _n}\left( {\lambda ,\B} \right) = \frac{3\B}{2\lambda J}$. Let us define the sequence,
\[
\zeta _{{\varepsilon _\gamma }}^n\left[ k \right] \DE \left( {{{\S}^2}{\Delta ^n}{\varepsilon _\gamma }} \right)\left[ k \right].
\]
Setting ${\rho _n}\left( {\lambda ,\B} \right) = \frac{3\B}{2\lambda J}$ in \eqref{eq:interbound} entails, 
\begin{align*}
&  \zeta _{{\varepsilon _\gamma }}^n\left[ 1 \right] - \zeta _{{\varepsilon _\gamma }}^n\left[ {J + 1} \right] \subseteq 2\lambda J\left[ {{\kappa _n} - \frac{{3{\beta _g}}}{{2\lambda J}},{\kappa _n} + \frac{{3{\beta _g}}}{{2\lambda J}}} \right] \hfill \\
   \Leftrightarrow & \ {\kappa _n} \in \frac{1}{{2\lambda J}}\left( {\left( {\zeta _{{\varepsilon _\gamma }}^n\left[ 1 \right] - \zeta _{{\varepsilon _\gamma }}^n\left[ {J + 1} \right]} \right) + \left[ { - 3{\beta _g},3{\beta _g}} \right]} \right).
\end{align*} 
For $J = {{6\B}}/{\lambda }$, the right hand side in the above is an interval of length ${1}/{2}$ and hence only contains one integer, namely, 
\begin{equation}
\label{eq:kn}
{\kappa _n} = \Big\lfloor {\frac{{\left( {{\S^2}{\Delta ^n}{\varepsilon _\gamma }} \right)\left[ 1 \right] - \left( {{\S^2}{\Delta ^n}{\varepsilon _\gamma }} \right)\left[ {J + 1} \right]}}{{12\B}} + \frac{1}{2}} \Big\rfloor. 
\end{equation}
Together with the rounding operation defined in \eqref{eq:RD}, the estimate of ${\kappa _n}$ allows us to recursively compute, 
\begin{equation}
\label{eq:sn}
n\in [0,N-2],\quad s_{\l n + 1 \r}[k] = 2\lambda \Big\lceil {\frac{{\left\lfloor {\S s_{\l n  \r}[k]/\lambda } \right\rfloor }}{2}} \Big\rceil+ 2\lambda\kappa_n
\end{equation}
where the sequence $s_{\l n  \r}[k]$ is initialized with initial condition $ s_{\l {0}  \r}[k] \EQc{eq:egy}\VOB{\gamma}[k]$.
This completes the recovery procedure which is summarized in Algorithm \ref{alg:ModSamp}. The following theorem provides a recovery guarantee for this algorithm.
\begin{theorem}[Unlimited Sampling Theorem] \label{thm:UST} Let $g\l t\r \in \PW{\Omega}$ and 
${y[k]} = {\left. {\MO{ {g\left( t \right)}}} \right|_{t = k\T}},k \in \mathbb{Z}$ in (\ref{yn}) be the modulo samples of $g\l t \r$ with sampling rate $\T$. Then a sufficient condition for recovery of $g\l t \r$ from the $\{y[k]\}_k$ {up to additive multiples of $2\lambda$} is that
\begin{equation}
\label{MSB}
\T \leq \frac{1}{2\Omega e}.
\end{equation}
\noindent Provided that this condition is met and assuming that $\B\in 2\lambda \mathbb{Z}$ is known with $\|g\|_\infty\leqslant \B$, then choosing
\begin{equation}
\label{NB}
N = \left\lceil {\frac{{\log \lambda  - \log {\B }}}{{\log \left( {{\T}\Omega e} \right)}}} \right\rceil,
\end{equation}
yields that ${\left( {\T\Omega e} \right)^N}{\left\| {g} \right\|_\infty } < \lambda$ and Algorithm 1 recovers $g$ from $y$ {again up to the ambiguity of adding multiples of $2\lambda$}. 
\end{theorem}
\ab{The proof of this Theorem is deferred to the next section (Section \ref{sec:RNC}), which considers the case of bounded noise. The same proof applies in the absence of noise.}

\subsubsection{Recovery Conditions for the Case of Bounded Noise}
\label{sec:RNC}

\ab{In this section, we consider the effect of bounded noise on the recovery procedure. More precisely, we assume that the modulo samples $y[k]$ are affected by noise $\eta$ of amplitude bounded by a constant $ \bo > 0$. That is,
\begin{equation}
\label{eq:YN}
\forall k \Z, \quad \YN\left[ k \right] = y[k] + \eta\left[ k \right], \quad \left| {\eta \left[ k \right]} \right| \leqslant {\bo}.
\end{equation}
Note that due the presence of noise, it may happen that $\YN [k] \not\in [-\lambda,\lambda]$. Nonetheless, for $\bo$ below some fixed threshold, our recovery method provably recovers noisy bandlimited samples $\gamma[k]$ from the associated noisy modulo samples $\YN[k]$ up to an unknown additive constant, where the noise appearing in the recovered samples is in entry-wise agreement with the one affecting the modulo samples. That is, $\widetilde \gamma \left[ k \right] = \gamma \left[ k \right] + \eta \left[ k \right] + 2m\lambda, m\Z$ .  
} 

\begin{theorem}[Unlimited Sampling Theorem with Noise] \label{thm:USTN} Let $g\l t\r \in \PW{\Omega}$ and assume that $\B\in 2\lambda \mathbb{Z}$ is known with $\|g\|_\infty\leqslant \B$. For the dynamic range we work with the normalization $\DR = {\B}/{\lambda}$.

Let the noisy modulo samples be of the form \eqref{eq:YN} with a noise bound given in terms of the dynamic range as 
\begin{equation}
\label{eq:mns}
\left\| \eta  \right\|_\infty \leqslant \frac{\lambda }{4}{\left( {{{2\cdot\DR}}} \right)^{ - \frac{1}{\alpha}}},  \qquad \alpha \in \mathbb{N}.
\end{equation}
Then a sufficient condition for Algorithm \ref{alg:ModSamp} with
\begin{equation}
\label{NB_noise}
N = \left\lceil {\frac{{\log \lambda  - \log {\B }}-1}{{\log \left( {{\T}\Omega e} \right)}}} \right\rceil,
\end{equation}
to approximately recover the bandlimited samples $\gamma[k]$ in the sense that it returns $\widetilde \gamma \left[ k \right] = \gamma \left[ k \right] + \eta \left[ k \right] + 2m\lambda, m\Z$ is that
\begin{equation}
\label{MSBN}
\T \leq \frac{1}{2^\alpha \Omega e}.
\end{equation}
\end{theorem}

\begin{proof}
 We start by observing that the modified choice  of $N$ as given in  \eqref{NB_noise} yields that ${(T\Omega e)^N} \leq {\left( {{\lambda }/{{{(2\beta _g)}}}} \right)}$. Together with \eqref{lem:1b}, this yields that $\|\Delta^N \gamma\|_\infty\leqslant \frac{\lambda }{2} $.
Similarly, the modified choice of $T$ yields that
\begin{equation*} 
{2^N} \leqslant 2\cdot 2^{{\frac{{\log \lambda  - \log {\B }}-1}{{\log \left( {{\T}\Omega e} \right)}}} }
\leqslant 
{2\cdot\left( {\frac{2\B }{\lambda}} \right)^{\frac{1}{{ \alpha }}}} = 
{2\cdot\left( {2\cdot\DR } \right)^{\frac{1}{{ \alpha }}}} ,
\end{equation*}
which yields via Young's inequality and \eqref{eq:mns} that 
\begin{equation}
\label{eq:2Ngam} 
\|\Delta^N \eta\|_\infty \leq 2^N \|\eta\|_\infty \leqslant  \frac{\lambda}{2}.
\end{equation}
These observations now entail that even the noisy modulo measurements $\YN$ allow for exact computation of $\VO{\gamma}[k]$. Indeed, we observe that
\begin{align*}
&\MO{\Delta^N\YN} - \Delta^N\YN  \\
& \DL{a} \MO{\MO{\Delta^N \eta} + \MO{\Delta^N y}}- \Delta^N \l y +  \eta\r \\
& \DL{b} \MO{\Delta^N \gamma + \Delta^N \eta}- \Delta^N y- \Delta^N \eta \\
& \DL{c} \Delta^N \gamma - \Delta^N y = \ND{\VO{\gamma}}.
\end{align*}  
Here in (a), we used Proposition~\ref{prp:COM}. Equality (b) follows from \eqref{eq:modeq2} and \eqref{eq:2Ngam}. For (c) we again use \eqref{eq:2Ngam} combined with the bound $\|\Delta^N \gamma\|_\infty\leqslant \frac{\lambda }{2} $ derived above.

We will show now by induction in $m \in [0, N-1]$ that ${s_{\left( m \right)}} = {\Delta ^{N - m}\VO{\gamma}}$ for ${s_{\left( m \right)}}$ as in \eqref{eq:sn} and $\VO{\gamma}$ as in \eqref{eq:epsgamma}. The choice $m=N-1$ then ensures that ${\S s_{\left( N-1 \right)}} = \VO{\gamma} + 2k\lambda$ for some $k\in \mathbb{Z}$. As per modulo decomposition in Proposition \ref{prop:mod}, this ensures that step 5 of Algorithm \ref{alg:ModSamp} returns \ab{noisy} bandlimited samples $\widetilde \gamma[k] =\gamma[k] + 2m \lambda + \eta[k]$, that is, one  recovers up to the \ab{measurement noise and an} unavoidable constant ambiguity of $2\lambda\mathbb{Z}$. The result then follows from Shannon's sampling theorem which is implemented in step 6 of Algorithm \ref{alg:ModSamp}. 

To prove that ${s_{\left( m \right)}} = {\Delta ^{N - m}\VO{\gamma}}, \quad m \in [0, N-1]$, we first note that the induction seed $m=0$ reduces to \ab{showing that} $s_{\l{0}\r} = \ND{\VO{\gamma}}$, which holds by definition. 

For the induction step, assume that for some $ m \in [0, N-2]$, one has ${s_{\left( m \right)}} = {\Delta ^{N - m}\VO{\gamma}}$. Recall that we have derived above that for $J = 6\B/\lambda$ as chosen in Algorithm \ref{alg:ModSamp}, $\kappa_m$ given by \eqref{eq:kn} is the unique multiple of $2\lambda$ for which $s_{\l m+1 \r}[k]$ defined in \eqref{eq:sn} is a bounded sequence. Hence, combining \eqref{eq:S1} with the induction hypothesis and the observation that ceiling and floor in \eqref{eq:sn} only have an effect in case of round-off errors, we conclude that ${s_{\left( {m + 1} \right)}} = {\Delta ^{N - \left( {m + 1} \right)}} \VO{\gamma}$ as desired.
\end{proof}

\begin{algorithm}[!t]
\label{alg:ModSamp}
\KwData{\ \ ${\YN[k]}$ and $2\lambda \mathbb Z \ni \B \geqslant {\left\| g \right\|}_\infty$.}
\KwResult{$\widetilde g\l t \r \approx g\l t \r$.}
\begin{enumerate}[leftmargin=*,label=$\arabic*)$]
\itemsep 2pt
\item Compute $N = \left\lceil {\frac{{\log \lambda  - \log \B}}{{\log \left( {\T\Omega e} \right)}}} \right\rceil$ using \eqref{eq:NO}.
\item Compute $\l \ND{\YN}\r [k]$.
\item Compute $\l\ND{\VO{\gamma}}\r[k] =  \l \MO{\ND{\YN}} -\ND{\YN}\r [k]$. \\ 
	Set $s_{\l 0 \r}[k] = \l\ND{\VO{\gamma}}\r[k]$.
\item for $n = 0:N-2$

\begin{enumerate}[leftmargin=*,label= $\l\roman*\r$]
\item $s_{\l n + 1 \r}[k] = \l\S s_{\l n \r}\r[k]$.

\item $s_{\l n + 1 \r} \EQc{eq:RD} 2\lambda \left\lceil {\frac{{\left\lfloor {s_{\l n + 1 \r}/\lambda } \right\rfloor }}{2}} \right\rceil$ (rounding to $2\lambda\mathbb{Z}$).

\item Compute $\kappa_{n}$ in \eqref{eq:kn}.

\item $s_{\l n + 1 \r}[k] = s_{\l n + 1 \r}[k] + 2\lambda \kappa_{\l n \r}  $.

\end{enumerate}
end 

\item $\widetilde \gamma [k] = \l \S s_{\l N-1 \r}\r [k] + \YN[k] + 2m\lambda,\ \   m\Z$.

\item Compute $\widetilde g\l  t\r$ from $\widetilde \gamma[k] $ using low-pass filtering,
\[\widetilde g\left( t \right) \EQc{ST} \sum\limits_{k \in \mathbb{Z}} {\widetilde \gamma } \left[ k \right]\operatorname{sinc} \left( {t/\T - k} \right).\] 
\end{enumerate}
\caption{Recovery from Noisy Modulo Samples}
\end{algorithm}

\begin{figure*}[!t]
\centering
\includegraphics[width =1\textwidth]{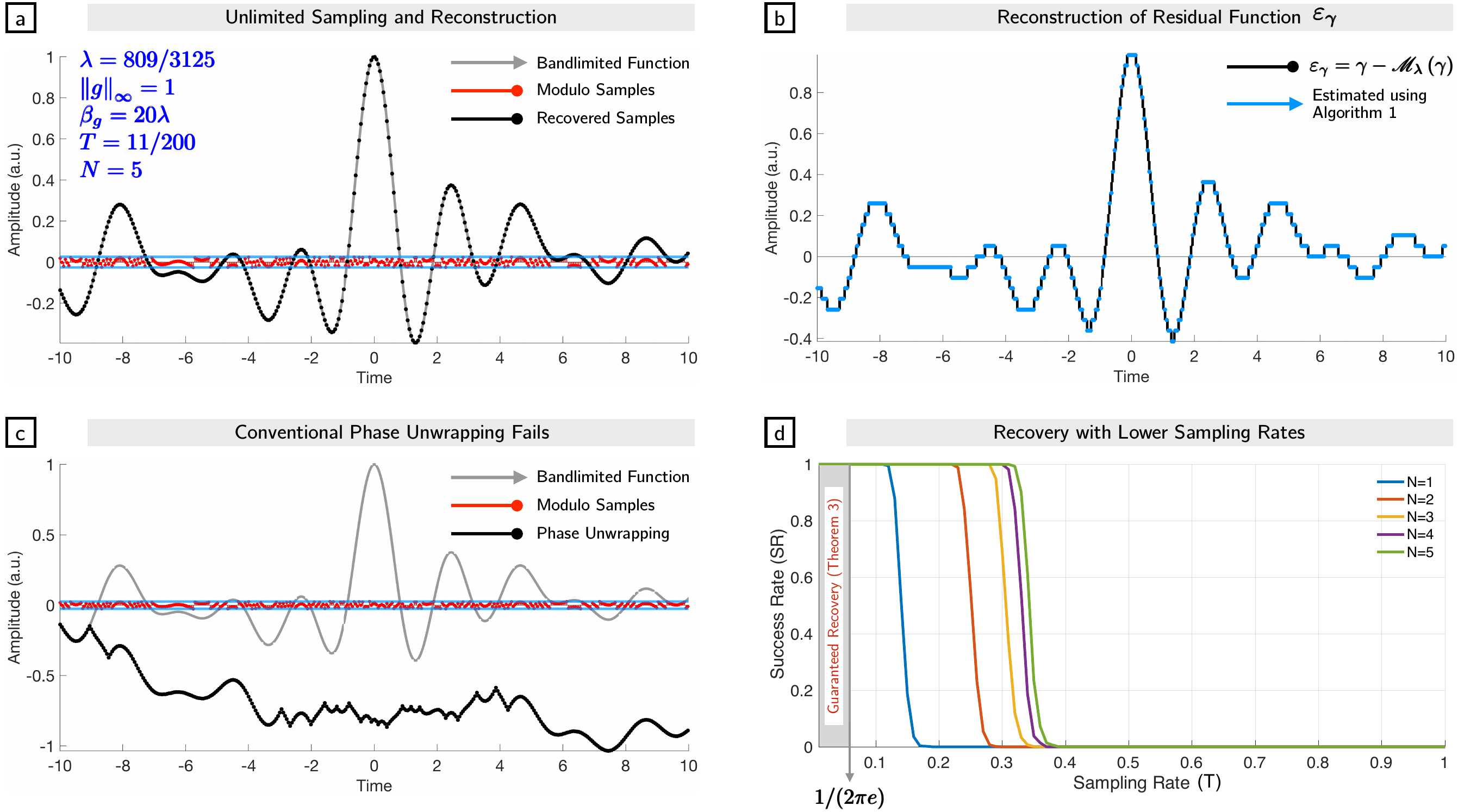}
\caption{Modulo sampling of bandlimited functions. (a) One of the $1000$ realizations of a randomly generated bandlimited function $\BLP{g}$, its modulo samples $y[k]$ and perfectly recovered samples $\widetilde{\gamma}[k]$. (b) Residual function $\VO{\gamma}$ and its estimated version which is recovered up to machine precision. (c) Comparision between Algorithm \ref{alg:ModSamp} and phase unwrapping (based on Itoh's condition \cite{Itoh:1982}). Phase unwrapping method fails to recover bandlimited samples. (d) Assessment of the sharpness of the unlimited sampling theorem. }
\label{fig:Result_1}
 \end{figure*}

\section{Numerical Experiments}
\label{sec:NE}

In this section, we numerically verify several results linked with the topic of unlimited sampling and reconstruction. In particular, (a) we verify our sampling theorem and explore its limitations, (b) compare Algorithm \ref{alg:ModSamp} with Itoh's method for phase unwrapping and (c) confirm that our approach is stable in the presence of quantization noise. These results are plotted in \fig{fig:Result_1} (a)--(d) and \fig{fig:Result_2} (a)--(c). 

\subsection{Verifying the Unlimited Sampling Theorem}

\label{sec:VUST}

In order to verify our sampling theorem, we generate $1000$ realizations of a bandlimited function $\BLP{g}$ with piecewise constant Fourier spectrum taking values chosen uniformly at random, $\widehat{g}\l \omega \r \in \ud{0}{1}$ where $\ud{a}{b}$ denotes the uniform random distribution. The bandlimited function is normalized such that $||g||_\infty =1$. Furthermore, for each realization, we use $\lambda \in \ud{\frac{1}{100}}{\frac{1}{10}}$ and $\B = \left\lfloor {{{\left\| g \right\|}_\infty }/2\lambda } \right\rfloor$. We obtain modulo samples $y[k]$ in \eqref{yn} using sampling rate $\T= \frac{11}{200}<\frac{1}{2\pi e}$. There on, we use Algorithm \ref{alg:ModSamp} to recover bandlimited samples. Although Algorithm \ref{alg:ModSamp} recovers bandlimited samples up to an unknown additive constant of $2\lambda\mathbb{Z}$, we use the knowledge of ground truth to estimate this unknown constant so that we can compute the mean squared error. For each of the $1000$ realizations, the samples $\widetilde{\gamma}$ are recovered up to machine precision ($10^{{-33}}$). 

A specific realization of the experiment is shown in \fig{fig:Result_1} (a). With $\lambda = 809/3125, \B = 20\lambda$ and $\T= {11}/{200}$, we estimate $N=5$. The recovered samples $\widetilde{\gamma}$ are also shown in \fig{fig:Result_1} (a). The mean squared error between the bandlimited samples and recovered samples in this case is $6.2455\times 10^{-33}$. The residue function corresponding to this realization, that is $\VO{\gamma}$, is computed using the ground truth and its estimated version which is obtained using Algorithm \ref{alg:ModSamp}. This is shown in \fig{fig:Result_1} (b). We also compare our reconstruction with the phase unwrapping method based on Itoh's condition \cite{Itoh:1982}. As explained in \fig{fig:PU}, Itoh's condition succeeds when $|\Delta y| < \lambda$, which is not the case here. As shown in \fig{fig:Result_1} (c), phase unwrapping based reconstruction fails.

 \subsection{Exploring Sharpness of the Unlimited Sampling Theorem}
\label{sec:VIUST}

On one hand, the injectivity condition in Theorem \ref{thm:IUS} proves that any sampling rate above the Nyquist rate guarantees unique representation of a bandlimited function in terms of its modulo samples. On the other hand, for Algorithm \ref{alg:ModSamp} to succeed, the sampling rate must be a factor of $2 \pi e$ faster than Nyquist rate. Here we set up a demonstration that shows that Algorithm \ref{alg:ModSamp} succeeds even when the sampling rate is much slower than what is prescribed by Theorem \ref{thm:UST}, however, the recovery is not guaranteed. For any $\BLP{g}$, we denote the critical sampling rate by $\TS = 1$ and the corresponding sampling rate prescribed by the unlimited sampling theorem is $\TUS = 1/2\pi e$. We consider modulo samples  acquired with $\lambda = 2/10$, so the theorem is valid for all orders $N\geq 1$. To assess the sharpness of Theorem \ref{thm:UST}, we use $1000$ realizations of randomly generated bandlimited functions with $||g||_\infty = 1$. Each realization is then sampled with sampling rate $\T = 0.01,\ldots,\TS$ in steps of $1/100$ and we set $\B=1$. For sampling rates beyond the applicability of the theorem, that is, $\T>\TUS$, we run Algorithm \ref{alg:ModSamp} for each of $1000$ realizations by varying from $N = 1$ to $N=5$. Then for each combination of $\T$ and $N$, we compute the \emph{success rate} which is defined as the fraction of trials in which the algorithm reconstructs up to machine precision (in the sense of mean-squared error). As shown in \fig{fig:Result_1} (d), indeed the recovery is possible even when $\T \in \l \TUS,\TS\r$. Furthermore, higher order differences extend recover to smaller over-sampling factors, but there seems to be a minimum oversampling rate around $\T=0.4$ above which the algorithm always fails for the chosen value of $\lambda$.

\subsection{Quantization and Bounded Noise}
\label{sec:VS}

As discussed in Section \ref{sec:RNC}, Theorem~\ref{thm:USTN} provides recovery guarantees in the presence of bounded noise. A main motivation for this noise model is that it includes the error caused by round-off quantization as in Section \ref{ssec:MMUS}, Step 5. In \fig{fig:Result_2} we illustrate the recovery guarantees for $3$--bit quantized modulo samples. That is, each modulo measurement is rounded to the closest element of
\[
\EuScript{C}_{3,\lambda} = \left\{\pm\frac{\lambda}{8},\pm\frac{3\lambda}{8},\pm\frac{5\lambda}{8},\pm\frac{7\lambda}{8}\right\}.
\] 
In other words, the quantized measurements $y_\eta$ are incurring quantization noise $\|\eta\|_\infty\leq \tfrac{\lambda}{8}$. In the experiment, we use a bandlimited function with $||g||_\infty = 12.5$ and $\lambda =1$. As before, $\T= {11}/{200}$ with $\B = 14\lambda$ and we set $N=2$. \aq{The experiment shows that our method is even more robust than predicted by Theorem~\ref{thm:USTN}:} despite the quantization noise $\left\| \eta  \right\|_\infty \approx \lambda/8$, Algorithm \ref{alg:ModSamp} recovers the underlying (noisy) bandlimited samples which are shown in \fig{fig:Result_2} (b). The reconstruction mean squared error (MSE) between $\gamma$ and its estimate $\widetilde{\gamma}$ (using Algorithm \ref{alg:ModSamp}) is is observed to be $5.1\times 10^{-3}$ which is the same as the  MSE between $y$ and $\YN$ and much smaller than the MSE of $3\times 10^{-1}$ incurred with $3$-bit quantization of non-modulo samples. 

\begin{figure}[!t]
\centering
\includegraphics[width =0.8\columnwidth]{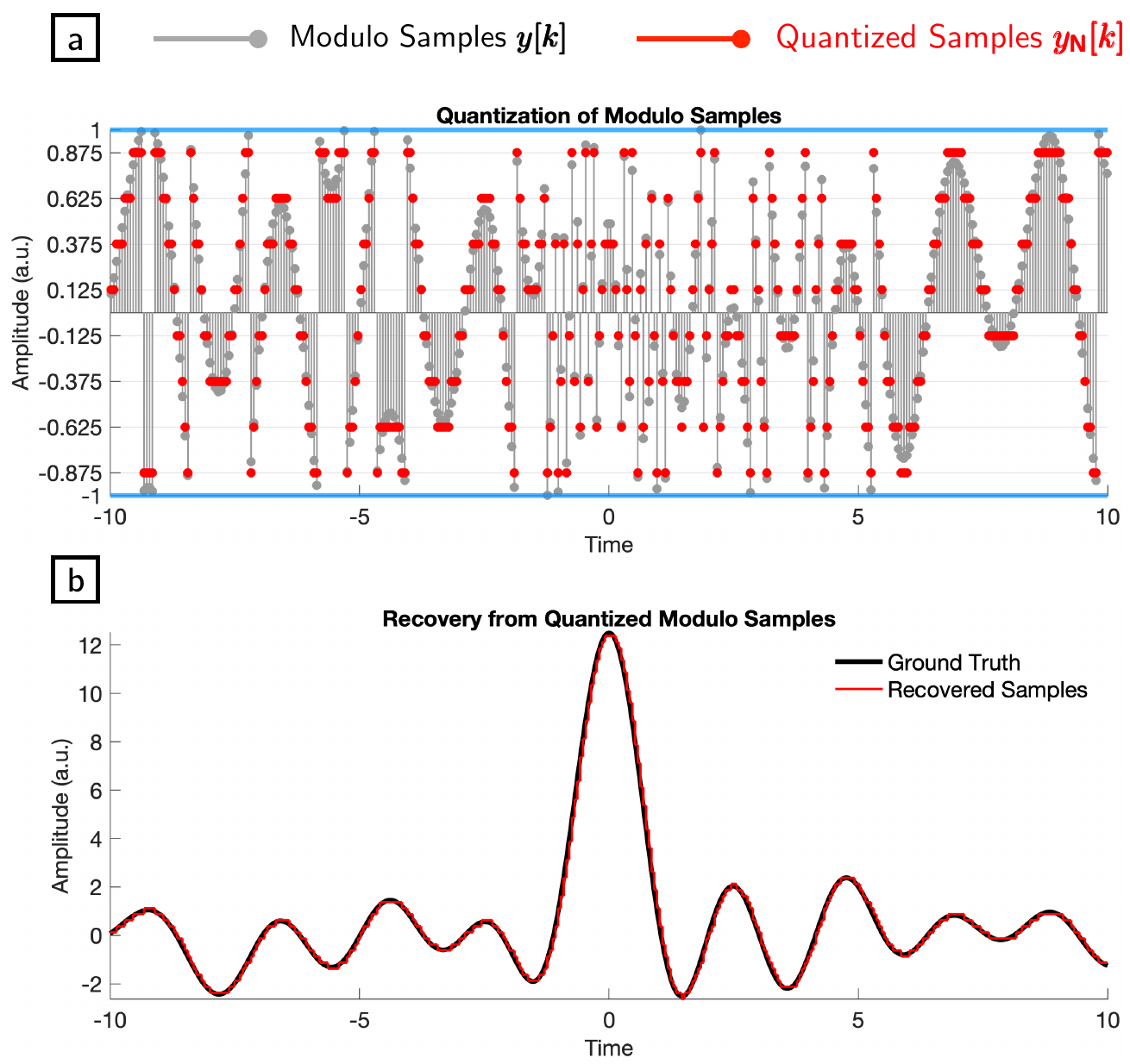}
\caption{Algorithm \ref{alg:ModSamp} is stable in the presence of bounded noise which is the model for quantization. (a) Modulo samples $y[k]$ and its $3$-bit quantized version $\YN\left[ k \right] = y[k] + \eta\left[ k \right]$. Here $\B = 12.5$, $\lambda = 1$ and$\left\| \eta  \right\|_\infty \leqslant \lambda/8$. (b) Ground truth and recovered modulo samples $\widetilde \gamma$ using Algorithm \ref{alg:ModSamp}.}
\label{fig:Result_2}
 \end{figure}

\section{Conclusion}

\label{sec:Conc}

\subsection{Summary of Results} In this work, we have described a new sensing and recovery scheme that overcomes the dynamic range barrier fundamental to digital sensing and sampling theory. Our approach harnesses a co-design between hardware and algorithms. On the hardware front, a modulo non-linearity injected in the sensing process scrambles high-dynamic-range information into low-dynamic range samples. Theorem \ref{thm:IUS} addresses the question of uniqueness; a bandlimited function is uniquely characterized by its modulo samples provided that the sampling rate is faster than critical sampling density. On the algorithmic front, given modulo samples of an $\Omega$-bandlimited function, the sampling rate $\T\leq 1/\l 2\Omega e\r$ guarantees recovery. This Unlimited Sampling Theorem is formally stated in Theorem \ref{thm:UST}; stability with respect to quantization noise and other bounded noise is established in Theorem \ref{thm:USTN}. The proposed reconstruction approach hinges on the observation that the modulo and the (higher-order) difference operators commute in a certain sense. 

\subsection{Future Directions}
\begin{enumerate}[leftmargin=*,label= $\bullet$]
\itemsep5pt
\item {\bf Implementation Issues}      \aq{Implementing our approach in practice remains an integral aspect of future research. In particular, we are working
to design computational techniques to address inaccuracies and instabilities arising in the implementation of our techniques in hardware. We are currently exploring such instabilities in the context of a prototype ADC. Important challenges include timing jitter and that our sensing pipeline requires ADCs without anti-aliasing filter, which makes them more sensitive to noise. For this reason, noise control needs to be implemented in alternative ways.}

  \item {\bf Wider Classes of Signals and Function Spaces} A natural question to ask is how our results can be extended to a wider class of signals such as smooth functions (shift-invariant spaces), sparse signals, parametric classes of signals among others. Some partial answers are provided in \cite{Rudresh2018,Cucuringu2018,Musa2018,Shah2018,Bhandari2018,Bhandari2018a} but this is an open area with several unanswered questions. Unlimited sampling in spline-spaces was recently discussed in \cite{Bhandari:2020a}. Another interesting variation is to consider multi-dimensional signal models defined on arbitrary lattices, such as \cite{Bouis:2020}.
  
  
\item {\bf Wider Classes of Inverse Problems} The unlimited sensing framework leads to a new class of inverse problems due to the modulo non-linearity. It can be further combined with several interesting inverse problems where dynamic range poses a natural limitation. More generally, one may consider $y_n =  \l\mathscr{M}_\lambda\circ\mathscr{T}\r \l x_n\r$ where $\mathscr{T}$ is some operator. For instance, using $\mathscr{T}$ as the Radon Transform, recently we have introduced the \emph{Modulo Radon Transform} in the context of high-dynamic-range computed tomography \cite{Bhandari:2020}.

\item {\bf New Sampling Architectures} The unlimited sensing framework can be combined with different sampling architectures. For instance, in \cite{Graf2019}, one-bit recovery is proposed based on the sigma-delta modulation scheme. Similarly, this work can be combined with multi-channel acquisition approaches which naturally find applications in imaging, for example, time-of-flight imaging \cite{Bhandari:2018} and our upcoming work on direction-of-arrival estimation from modulo samples \cite{FernandezMenduina:2020}. 
    
\item {\bf Robustness to Outliers} While our results yield recovery for  bounded noise, our approach is not yet stable with respect to outliers. Given the redundancy of the data, we expect that it will be possible to address this issue and, subsequently, study unbounded noise models such as Gaussian and Poissonian noise.

\end{enumerate}

\ifCLASSOPTIONcaptionsoff
\newpage
\fi

%
%
%

\bibliographystyle{IEEEtran_url}
\bibliography{IEEEabrv,US_TSP}

\end{document}